\documentclass[11pt]{article}

\usepackage{enumitem}
\usepackage{geometry}
\usepackage{color}
\usepackage{amsmath}
\usepackage{amssymb}
\usepackage{amsthm}
\usepackage{amsfonts}
\usepackage{graphicx}
\usepackage{xspace}
\usepackage{url}
\usepackage{bbm}

\usepackage{enumitem}

\usepackage{color}

\def\E{\mathbb{E}}

\def\dalg{\mu_{\textsc{alg}}^*}
\def\dalgOverall{\mu_{\textsc{alg}}}
\def\fdalg{\tilde{\mu}_{\textsc{alg}}}
\def\poly{\mathrm{poly}}

\def\E{\mathbb{E}}

\def\eps{\epsilon}

\def\lam{\lambda}

\def\cC{\mathcal C}

\def\cI{\mathcal I}

\def\cL{\mathcal L}
\def\cS{\mathcal S}

\newtheorem*{theorem*}{Theorem}
\newtheorem{theorem}{Theorem}[section]
\newtheorem{lemma}[theorem]{Lemma}
\newtheorem{corollary}[theorem]{Corollary}

\newtheorem{definition}{Definition}

\newtheorem*{prop*}{Proposition}

\newtheorem*{claim*}{Claim}
\newtheorem*{fact*}{Fact}
\newtheorem*{remark*}{Remark}
\newtheorem*{defn*}{Definition}

	\title{Fast and perfect sampling of subgraphs and polymer systems}
	\author{Antonio Blanca\footnote{Department of Computer Science and Engineering, Pennsylvania State University. ablanca@cse.psu.edu.}  
		\and Sarah Cannon\footnote{Department of Mathematical Sciences, Claremont McKenna College.  scannon@cmc.edu.} 
		\and Will Perkins\footnote{Department of Mathematics, Statistics, and Computer Science, UIC. math@willperkins.org.}}

	\date{\today}
	
\begin{document}

\maketitle
	
	\begin{abstract} 
	We give an efficient perfect sampling algorithm for weighted, connected induced subgraphs (or \textit{graphlets}) of  rooted, bounded degree graphs. 
	Our algorithm utilizes a vertex-percolation process with a carefully chosen rejection filter and works under a percolation subcriticality  condition. 
	We show that this  condition is optimal in the sense that the task of (approximately) sampling weighted rooted graphlets  becomes impossible in finite expected time for infinite graphs and intractable for finite graphs when the condition does not hold.  We apply our  sampling algorithm as a subroutine to give near linear-time perfect sampling algorithms for polymer models and weighted non-rooted graphlets in finite graphs, two widely studied yet very different problems.  
	This new perfect sampling algorithm for polymer models gives improved sampling algorithms for spin systems at low temperatures on expander graphs and unbalanced bipartite graphs, among other applications. 
\end{abstract}

\section{Introduction}

Sampling is a fundamental computational task: given a specification of a probability distribution on a (large) set of combinatorial objects, output a random  object with the specified distribution or with  a distribution close to the specified distribution.    This task becomes challenging when the specification of the distribution is much more succinct than the set of objects, and one wants to sample using time and space commensurate with the specification.  Fundamental examples include sampling from Markov random fields and probabilistic graphical models and sampling substructures of graphs.  We will address both of these examples here and connect them in a new way.  	

We consider a natural sampling problem: given a bounded-degree graph $G$, sample a \emph{graphlet} (a connected, vertex-induced subgraph) of $G$ containing a fixed vertex $r $ with probability proportional to an exponential in the size of the subgraph.~That is, sample a graphlet $S$ containing vertex $r$ with probability proportional to $ \lambda^{|S|}$, where $\lambda > 0$ is a distribution parameter and $| S | $ denotes the number of vertices in~$S$. 
In the paper we are concerned with small values of $\lambda$, where the expected size of a sampled graphlet is much smaller than the size of the graph.

Sampling graphlets  is an important task in data science, network analysis, bioinformatics, and sociology, 
as it allows us to gain information about massive graphs from small sections of it; see, e.g.,~\cite{ds1,ds2,ds3,ds4}.
A number of variants of the problem have consequently been studied, including sampling graphlets of a given size uniformly at random or sampling weighted graphlets of all sizes~\cite{shervashidze2009efficient,bhuiyan2012guise,jha2015path,chen2016general,bressan2017counting,bressan2018motif,paramonov2019estimating,agostini2019mixing,matsuno2020improved,read2021hardness,bressan2021faster,bressan2021efficient}.  
The variant we consider here, i.e., sample a graphlet $S$ with probability proportional to  $ 
\lambda^{|S|}$, 
arises as a key subroutine in recent sampling algorithms for spin systems (hard-core model, Ising model, Potts model, etc.)~in the regime of strong interactions via \textit{polymer models} described below in Section~\ref{secPolymerIntro}; see~\cite{helmuth2020algorithmic,cannon2020unbalanced,liao2019counting,galanis2021unbounded,borgs2020efficient,jenssen2021approximately,jenssen2020algorithms,carlson2020efficient,chen2022sampling}. 

One major limitation of previous sampling algorithms for graphlets and polymer models (those in, e.g.,~\cite{helmuth2020algorithmic,liao2019counting,read2021hardness,chen2021fast,galanis2021fast}, among others) is the use of exhaustive enumeration of graphlets of a given size; this requires restrictive parameter regimes or large polynomial running times, with the logarithm of the maximum degree $\Delta$ of the graph appearing in the exponent of the polynomial.
Here we design  a fast perfect sampling algorithm for weighted graphlets based on a vertex percolation process combined with a rejection filter. 
This method bypasses the enumeration barrier and allows us to design perfect sampling algorithms for a number of applications, substantially improving upon existing algorithms in three ways: 1) our algorithms have considerably faster running times, with no dependence on $\Delta$ in the exponent; 2) our algorithms return perfect, rather than approximate, samples from the desired distributions; and 3) our algorithms are conceptually simple and practical to implement. 

Our algorithm proceeds as follows.
First, run a vertex percolation process on the graph $G$ beginning at vertex $r$ in a breadth-first search manner, repeatedly adding each adjacent vertex to the graphlet with a carefully-chosen probability $p$.  Once the percolation process terminates, the graphlet is accepted as the random sample with a certain probability that depends on the graphlet and rejected otherwise;
if the graphlet is rejected, the algorithm restarts another percolation process from $r$. 
Because of the careful way we choose the percolation and rejection probabilities, we can prove the final accepted sample is drawn exactly from the desired distribution and the expected running time is bounded by a constant that depends only on $\lambda$ and the maximum degree $\Delta$. 

For our applications to polymer models, 
we use this graphlet sampling algorithm as a subroutine to implement a Markov chain on polymer configurations. We then use this Markov chain to devise a perfect sampling algorithm for polymer models
based on the coupling from the past method from~\cite{ProppWilson} and the notion of bounding chains from~\cite{huber2004bounding}.

\subsection{Sampling rooted graphlets}
\label{secIntroRooted}

Our key contribution is a new algorithm for perfectly sampling weighted graphlets containing a given vertex $r$.
We start by fixing the model of computation we work with throughout the paper. We assume a model that allows for querying the adjacency list of a given vertex
in a bounded degree graph in constant time, including in a rooted infinite graph.
We also assume that in a finite graph we can query a uniformly random vertex in constant time. This a standard model used in the study of sublinear algorithms~\cite{goldreich1997property}.
We also assume access to a stream of perfectly random real numbers in $[0,1]$.
The model of computation is fixed for consistency; in particular, our methods extend to other models, only requiring to adjust the running time to account for any additional computational overhead.

Let $G = (V,E)$ be a finite or infinite graph of maximum degree $\Delta$. For $r \in V$, let $\cS(G,r)$ be the set of all connected, vertex-induced subgraphs of $G$ containing $r$.
(The subgraph induced by $U \subseteq V$ has vertex set $U$ and includes all the edges of $G$ with both endpoints in $U$.)
We call $r$ the root of $G$ and the elements of $\cS(G,r)$ graphlets rooted at $r$.
For $\lam>0$ define the probability distribution $\nu_{G,r,\lam}$ on $\cS(G,r)$ by
\begin{equation} 
	\label{eq:rooted:dist}
	\nu_{G,r,\lam} (S) = \frac{ \lam^{|S|}}{ Z_{G,r,\lam}  }, 
	\quad \text{where} \quad Z_{G,r,\lam} = \sum_{\hat S \in \cS(G,r)} \lam^{|\hat S|} \,.
\end{equation}
The distribution is well defined when the normalizing constant $Z_{G,r,\lam}$, known as the partition function, is finite.
This is the case for every graph of maximum degree $\Delta$ and every $r$ when $\lam$ is below the critical threshold:
\begin{equation}
	\label{eqCrtiThresh}
	\lam_*(\Delta) = \frac{(\Delta-2)^{\Delta-2}}{(\Delta-1)^{\Delta-1}} \,;
\end{equation}
see Lemma~\ref{lemma:part:finite} below.
This threshold was already considered in~\cite{read2021hardness}, who
provided an $\varepsilon$-approximate sampling algorithm for $\nu_{G,r,\lam}$ for the class of maximum-degree $\Delta$ graphs when $\lam < \lam_*(\Delta)$
with running time $\mathrm{poly}(\varepsilon^{-1})$.
We give a perfect sampling algorithm for $\nu_{G,r,\lam}$ for $\lam < \lam_*(\Delta)$ with constant expected running time. 

\begin{theorem}\label{thmSubgraphSimple}
	Fix $\Delta \ge 3$ and let $\lam < \lam_*(\Delta)$.  
	There is a randomized algorithm
	that for any graph $G=(V,E)$ of maximum degree $\Delta$ and any $r \in V$ 
	outputs a graphlet distributed according to $\nu_{G,r,\lam}$ with expected running time bounded by a constant that depends only on $\Delta$ and~$\lam$. 
\end{theorem}

\noindent
Previous algorithms to generate $\varepsilon$-approximate samples from $\nu_{G,r,\lam}$ (e.g., those in~\cite{helmuth2020algorithmic,read2021hardness,chen2021fast,galanis2021fast})
exhaustively enumerate all
graphlets of size $\le k$, for some $k$ that depends on the error parameter $\varepsilon$ that describes how accurate the sample must be.
This results in algorithms with $(1/\varepsilon)^{O(\log \Delta)}$ running times.
Applications such as sampling from polymer models require 
multiple samples from $\nu_{G,r,\lam}$ and have small error tolerance per sample; in particular, they require $\varepsilon \ll 1/n$,
which results in inefficient algorithms with overall running time $n^{O(\log \Delta)}$.
The algorithm in Theorem~\ref{thmSubgraphSimple}, on the other hand, is an exact sampler whose expected running time depends only on $\Delta$ and $\lambda$ 
and thus provides a significant advantage in applications as we detail~below.

We also show that Theorem~\ref{thmSubgraphSimple} is sharp in two ways. 
First, we establish that there is no polynomial-time approximate sampling algorithm for $\nu_{G,r,\lam}$ when $\lam \in (\lam_*(\Delta),1)$ for the class of finite graphs of maximum degree at most $\Delta$
unless RP=NP;
see Definition~\ref{def:approx:sampler} for the precise definition of a polynomial-time approximate sampler.
(We note that a similar hardness result was proved in~\cite{read2021hardness} for the related problem of sampling ``unrooted graphlets''; we provide more details about this in Section~\ref{secUnRootedIntro}.)
Second, 
in the infinite setting, the normalizing constant $Z_{G,r,\lam}$ may diverge (and consequently the distribution $\nu_{G,r,\lam}$ is not be well-defined)
when $\lam > \lam_*(\Delta)$; conversely, we prove that $Z_{G,r,\lam}$ is finite on every graph of maximum degree $\Delta$ when $\lam \le \lam_*(\Delta)$.

\begin{lemma}
	\label{lemma:comp:hard}
	Fix $\Delta \ge 3$ and $\lambda \in (\lambda_*(\Delta),1)$.
	If there is a polynomial-time approximate sampler for $\nu_{G,r,\lam}$
	for finite graphs $G=(V,E)$ of maximum degree $\Delta$ and each $r \in V$,
	then~RP=NP.
\end{lemma}

\begin{lemma}
	\label{lemma:part:finite}
	The partition function $Z_{G,r,\lambda}$ is finite for every (possibly infinite) graph $G=(V,E)$
	of maximum degree $\Delta$
	and every $r \in V$ if and only if $\lambda \le \lambda_*(\Delta)$.
\end{lemma}

\noindent
Finally, we mention that the algorithmic result in Theorem~\ref{thmSubgraphSimple} cannot be extended even to the case $\lam = \lam_*(\Delta)$: for the infinite $\Delta$-regular tree, we can show that the expected size of a graphlet sampled from $\nu_{G,r,\lam}$ is infinite when $\lam = \lam_*(\Delta)$, and so it is impossible
to have sampling algorithms with finite expected running time.
In summary, 
the algorithm in Theorem~\ref{thmSubgraphSimple} for $\lambda < \lambda_*(\Delta)$, combined 
with the hardness/impossibility results in Lemmas~\ref{lemma:comp:hard} and~\ref{lemma:part:finite} for $\lambda > \lambda_*(\Delta)$, provide a resolution to the computational problem of sampling from~$\nu_{G,r,\lam}$ on graphs of maximum degree at most $\Delta$.

As mentioned, our sampling algorithm is based on exploring the connected component of $r$ in a vertex-percolation process. We carefully choose a specific percolation parameter $p \in (0,1)$ as a function of $\lam$ and $\Delta$ (see Lemma~\ref{lemma:binary:search}).  We then perform breadth-first search (BFS) from $r$, labeling each new vertex encountered `active' with probability $p$ and `inactive' with probability $1-p$ independently over all vertices; we continue the BFS exploring only unexplored neighbors of active vertices.  In this way we uncover the `active' component of $r$, call it $\gamma$.  We then accept $\gamma$ with a given probability depending on $\lam$, $\Delta$, $|\gamma|$ and $|\partial \gamma|$, where $\partial \gamma$ denotes the set of vertices outside of $\gamma$ that are adjacent to $\gamma$.
If we reject $\gamma$, we begin again with a new percolation process.
We note that only when $\lam < \lam_*(\Delta)$ there exists a suitable percolation probability $p$ that results in a subcritical percolation process, so that the size of the active component has finite expectation and exponential tails. 
The weighted model we sample from is particularly well suited to this type of exploration algorithm because of the direct connection to a subcritical percolation~process.

Random exploration and rejection sampling have been used previously to sample graphlets and other structures, most notably in the recent work of Bressan~\cite{bressan2021efficient} who uses a novel bucketing scheme in combination with rejection sampling to perfectly  sample uniformly random graphlets of size $k$ from a graph, as well as studying the mixing time of the random walk on the set of all such graphlets.  See also~\cite{anand2021perfect} in which a random growth process and rejection sampling are used to perfectly sample spin configurations.  

We prove a more general version of Theorem~\ref{thmSubgraphSimple} in Section~\ref{secSubgraphs}, allowing for vertex-labeled graphlets and modifications of the weights by multiplication by a non-negative function bounded by $1$.  These generalizations are needed for the application to polymer models in Section~\ref{secPolymerIntro}.

\subsection{Sampling from polymer models}
\label{secPolymerIntro}

We use our algorithm for sampling weighted rooted graphlets to
design fast and perfect samplers for polymer models.
Polymer models are systems of interacting geometric objects representing defects from pure ground states (i.e., most likely configurations) in spin systems on graphs in classical statistical physics~\cite{gruber1971general,kotecky1986cluster,friedli2017statistical}.  These geometric objects are most often represented by vertex-labeled graphlets from a given host graph.
Recently, polymer models have found application as an algorithmic tool to sample from spin systems on various classes of graphs in strong interaction regimes; see, e.g.,~\cite{helmuth2020algorithmic,cannon2020unbalanced,liao2019counting,chen2021fast,helmuth2020finite,galanis2021fast,galanis2021unbounded,borgs2020efficient,jenssen2021approximately,jenssen2020algorithms,carlson2020efficient,coulson2020statistical,chen2022sampling}.
In these applications, the problem of sampling weighted vertex-labeled rooted graphlets emerged as a significant computational barrier.

We will work with \emph{subset polymer models} in which all polymers are vertex-labeled graphlets from a host graph $G=(V,E)$. These models were defined in~\cite{gruber1971general} and generalized in~\cite{kotecky1986cluster}. Such a polymer model consists of: 
\begin{itemize}
	\setlength{\itemsep}{0pt}
	\item A set $\cC=\cC(G)$ of polymers, each of which is a graphlet in $G$ with the vertices of the graphlet labeled with colors from a set $\Sigma$ of size $q$. 
	\item Weights $w_{\gamma} \ge 0$ for each $\gamma \in \cC$. We  assume without loss of generality that 
	all vertex-labeled graphlets of $G$, including
	each individual vertex $v\in V$, are elements of $\cC$, by setting $w_\gamma = 0$ when necessary. 
	\item 	An incompatibility relation $\nsim$ defined by connectivity.  We say two polymers $\gamma, \gamma' \in \cC$ are \textit{incompatible} and write $\gamma \nsim \gamma'$ if the union of their corresponding vertices induces a connected subgraph in $G$.  Otherwise they are \textit{compatible} and we write $\gamma \sim \gamma'$. 
\end{itemize}

Let $\Omega(\cC)$ denote the set of all sets of pairwise compatible polymers from $\cC$.
The polymer model is the Gibbs distribution $\mu$ on $\Omega(\cC)$ defined by
\[ \mu(X) = \frac{ \prod_{\gamma \in X} w_{\gamma}    }{ Z(\cC)} \, , \quad \text{where} \quad Z( \cC) = \sum_{X \in \Omega( \cC) } \prod_{\gamma \in X} w_{\gamma} \]
is the polymer model partition function. 
The size $|\gamma|$ of a polymer $\gamma$ is the number of vertices in the corresponding graphlet. 
We let $\cC_v$ be all polymers containing vertex $v$. 

We will often work with a family of polymer models corresponding to an infinite family of host graphs $G$.  We say the weights of a family of polymer models are \textit{computationally feasible} if $w_\gamma$ can be computed in time polynomial in $|\gamma|$ uniformly over the polymer models in the family.
In addition, we call
a polymer model \emph{$\delta$-covering} if for every vertex $v$, there exists a polymer $\gamma \ni v$ with weight $w_\gamma \ge \delta$.
We note that 
the $\delta$-covering assumption is very mild and is satisfied by essentially all standard polymer models arising in statistical physics. In particular, it fails only in degenerate situations where all polymers containing some vertex have vanishingly small or zero weight.

Algorithms for sampling polymer models fall into two classes: those based on truncating the  cluster expansion of a polymer model to approximate a partition function and using self-reducibility to sample, and those based on Markov chains on the set of collections of compatible polymers. The cluster expansion approach, while giving polynomial-time algorithms, is relatively inefficient in general, with the degree of the polynomial bound on the running time growing with the degree of the underlying graph; e.g., running time $n^{O(\log \Delta)}$ for $n$-vertex graph of maximum degree $\Delta$.  A Markov chain approach based on adding and remove polymers from a polymer configuration in principle can be much faster (near linear time in the size of the graph) but runs into one hitch: a much stricter condition on the parameters of the model is needed to perform one update step of the Markov chain efficiently (the ``polymer sampling condition'' in~\cite{chen2021fast,galanis2021fast}).  
We solve this problem by adapting our rooted graphlet sampler to sample polymers models, leading to a near linear-time perfect sampling algorithm for polymer models under the least restrictive conditions known (see Table~\ref{table:conditions}). 

\begin{theorem}
	\label{thmPolymerSample1}  
	Fix $\Delta \ge 3$, $q \ge 1$, $\theta \in (0,1)$, $\delta > 0$, and $\lambda < \lambda_*(\Delta,q) := \frac{(\Delta-2)^{\Delta-2}}{q(\Delta-1)^{\Delta-1}}$.
	There is a perfect sampling algorithm for $\mu$ with expected running time $O(|V| \log |V|)$ for any family of 
    $\delta$-covering
	subset polymer models on maximum degree $\Delta$ graphs $G = (V,E)$
	with  computationally feasible weights
	satisfying:
	\begin{equation}
		\label{eqPolyLamCondition}
		w_{\gamma} \le \lam ^{|\gamma|}~\text{for all}~\gamma \in \cC;~\text{and}
	\end{equation}
	\begin{align}
		\label{eqPerfectSamplingCondition}
		{\sum\nolimits_{\gamma \not\sim {v}} |\gamma| w_\gamma \le \theta~\text{for all}~v \in V}. 
	\end{align}
	
\end{theorem}	
The threshold defined in~\eqref{eqPolyLamCondition} is the generalization of the critical threshold for rooted graphlet sampling to the labeled case (taking $q=1$ recovers the  definition in~\eqref{eqCrtiThresh}).
Theorem~\ref{thmPolymerSample1} improves upon the known results for sampling from polymer models in two ways. For a very general class of polymer models, 
our algorithm simultaneously provides \emph{perfect sampling} with \emph{near-linear running time} under weak conditions on the polymer weights.  We now review previous works to illustrate these improvements; see the accompanying Table~\ref{table:conditions}.

\begin{table}
	\begin{center}
		\begin{tabular}{||c | c | c | c||} 
			\hline
			Condition & Bound on exponential  & Running & Precision \\
			&decay of weights&Time&\\ [1ex] 
			
			\hline
			Koteck{\`y}--Preiss~\cite{helmuth2020algorithmic,jenssen2020algorithms,cannon2020unbalanced} & $w_\gamma \le {(e^2 q \Delta )^{-|\gamma|}}$ & $n^{O(\log \Delta)}$ & approximate \\ [1ex]
			\hline
			Polymer Sampling~\cite{chen2021fast,galanis2021fast} & $w_\gamma \le {(e^5 q^3 \Delta^3 )^{-|\gamma|}}$ & $O(n \log n)$ & approximate \\[1ex]
			\hline
			Clique dynamics~\cite{friedrich2020polymer} & $w_\gamma \le {(e q \Delta )^{-|\gamma|}}$ & $n^{O(\log \Delta)}$ & approximate\\[1ex]
			\hline
			This work (Theorem~\ref{thmPolymerSample1}) & $w_\gamma \le {(e q \Delta )^{-|\gamma|}}$ & $O(n \log n)$ & perfect \\  \hline
		\end{tabular}
		\caption{Comparison of conditions and running times of known polymer sampling algorithms.}
		\label{table:conditions}
	\end{center}
\end{table}

A number of conditions on polymer weights have been used to provide efficient sampling algorithms.  The first papers in this direction (including~\cite{helmuth2020algorithmic,jenssen2020algorithms,cannon2020unbalanced}) used the Koteck{\`y}--Preiss condition for convergence of the cluster expansion of the polymer model partition function~\cite{kotecky1986cluster}: 
$ 
\sum_{\gamma' \nsim \gamma}  w_{\gamma'} e^{|\gamma'|}  \le |\gamma|~~\forall \gamma \in \cC.
$
This condition is typically verified by ensuring that:
\begin{equation}
	\label{eqKPcond}
	\sum\nolimits_{\gamma \nsim v}  w_{\gamma} e^{|\gamma|}     \le 1 \quad \forall v\in V\, .
\end{equation}
Since the number of vertex-labeled rooted graphlets of size $k$ in a maximum degree $\Delta$ graph grows roughly like $(e q\Delta)^{k-1}$~(see~\cite{BCKL}), weights of polymers of size $k$ must decay roughly like $(e^2 q \Delta)^{-k}$ for the polymer model to satisfy~\eqref{eqKPcond}, with the extra factor of $e$ coming from the exponential in the left hand side of the condition~\eqref{eqKPcond}. 

The  major downside to the algorithms based on the cluster expansion, i.e., those using  ~\eqref{eqKPcond} or the Koteck{\`y}--Preiss condition, is that the running times obtained are of the form $n^{O(\log \Delta)}$.  Subsequent works, namely~\cite{chen2021fast,galanis2021fast}, addressed this downside but at the cost of a significantly stricter condition on the polymer weights. 

In~\cite{chen2021fast}, the authors devised a new Markov chain algorithm for sampling from polymer models.  The condition on the polymer weights for rapid mixing of this chain is somewhat less restrictive than the Koteck{\`y}--Preiss condition; it is the Polymer Mixing condition:
\begin{equation} \label{eqn:polymermixing}
	\sum\nolimits_{\gamma' \nsim \gamma} |\gamma'| w_{\gamma'} \le \theta |\gamma| \, \quad \forall \gamma \in \cC~\text{for some } \theta \in (0,1) \,.
\end{equation}
This requires weights of polymers of size $k$ to decay like $(e q \Delta)^{-k}$, a savings of a factor $e$ in the base of the exponent over~\eqref{eqKPcond}.  However, to implement a single step of this Markov chain in constant expected time, a stronger condition (the Polymer Sampling condition) was required:  
\begin{equation}
	\label{eqPolySampling} 
	w_{\gamma} \le  \left (e^5 \Delta^3 q^3 \right)^{-|\gamma|} \,.
\end{equation}
This is a significant loss of a factor $e^3 \Delta^2 q^2$ in the base of the exponent compared to~\eqref{eqKPcond}, but the resulting sampling algorithm does run in near linear time.  

In~\cite{friedrich2020polymer}, the authors use a different Markov chain condition, the Clique Dynamics condition, similar to~\eqref{eqn:polymermixing},  which requires weights of polymers of size $k$ to decay like $(e q \Delta)^{-k}$, saving the same factor $e$ over~\eqref{eqKPcond}.  Their running times, though, are again of the form $n^{O(\log \Delta)}$ since implementing one step of their Markov chain involves enumerating rooted polymers of size $O(\log n)$.

Our results are a ``best-of-both-worlds'' for polymer sampling: under the conditions~\eqref{eqPolyLamCondition} and~\eqref{eqPerfectSamplingCondition} that both require polymer weights to decay like $(e q \Delta)^{-k}$ (this is shown later; see, e.g., the proof of Corollary~\ref{corPottsEx}), 
we obtain a near linear time algorithm. 
Moreover, unlike any of the previous results, our algorithm is a perfect sampler.  

To conclude this section, we comment briefly on the algorithm we design to sample from $\mu$.
Our starting point is the polymer dynamics Markov chain from~\cite{chen2021fast}.
We use it to implement a Coupling from the Past (CFTP) algorithm (see~\cite{ProppWilson}). To do so efficiently (in terms of the number of steps of the Markov chain), we design a new ``bounding Markov chain'' for the polymer dynamics, a method pioneered in~\cite{huber2004bounding,HN}, and
to implement each step of the Markov chain efficiently, we turn to our sampler for weighted rooted graphlets from Theorem~\ref{thmSubgraphSimple}.

\subsection{Applications to spin systems}
\label{sub:app:spin}

Our new algorithm for sampling subset polymer models
can be used as a subroutine in essentially all previous applications of polymer models for spin system sampling at low temperatures, including those in~\cite{jenssen2020algorithms,cannon2020unbalanced,liao2019counting,chen2021fast,helmuth2020finite,galanis2021fast,galanis2021unbounded,carlson2020efficient,coulson2020statistical,chen2022sampling}. This results in faster sampling algorithms under less restrictive conditions on model parameters in all those settings.
As examples, we fleshed out here the details in two of these applications; more details are provided in  Section~\ref{secApplications}.

\medskip\noindent\textbf{Hard-core model on bipartite  graphs.}	\ The {hard-core model} on a graph $G$ is the probability distribution $\mu_G^{hc}$ on $\cI(G)$, the set of all independent sets of $G$, with
\begin{equation}
	\label{eq:hc:dist}
	\mu_G^{hc}(I) = \frac{ \lam^{|I|}}{Z_G^{hc}(\lam)},\quad where \quad Z_G^{hc}(\lam) = \sum_{I \in \cI(G)} \lam^{|I|} \,. 
\end{equation}
The complexity of approximate counting and sampling from $\mu_G^{hc}$ on bounded-degree graphs is well understood: there is a computational threshold at some $\lam_c(\Delta)$, with efficient algorithms for $\lam<\lam_c(\Delta)$~\cite{weitz2006counting,ALO,ChenLiuVigoda,ChenLiuVigoda1} and hardness above the threshold (no polynomial-time algorithms unless NP=RP) \cite{sly2010computational,galanis2016inapproximability,sly2012computational}. However, on bipartite graphs, the complexity of these problems is unresolved and is captured by the class \#BIS (approximately counting independent sets in bipartite graphs) defined by Dyer, Goldberg, Greenhill, and Jerrum~\cite{dyer2004relative}. 

Theorem~\ref{thmPolymerSample1} implies the existence of a fast perfect sampling algorithm for the hard-core model in a certain class of bipartite graphs called {\it unbalanced bipartite graphs}, considered in~\cite{cannon2020unbalanced,friedrich2020polymer}.

\begin{corollary}
	\label{corHCunbalanced}
	There is a perfect sampling algorithm for $\mu_G^{hc}$ running in expected time $O(n \log n)$ for 
	$n$-vertex bipartite graphs $G$ with bipartition $(L,R)$, with  maximum degree $\Delta_L$ in $L$, maximum degree $\Delta_R$ in $R$, and  minimum degree $\delta_R$ in $R$ if
	\begin{equation}
		\label{eqn:unbalancedcond}
		\lambda(1 + (1+e) (\Delta_L-1)\Delta_R) < (1+\lambda)^{\delta_R/\Delta_L}.
	\end{equation}
\end{corollary}

Approximate sampling algorithms with large polynomial run times were previously given for this problem when ${6\lambda\Delta_L\Delta_R < (1+\lambda)^{\delta_R/\Delta_L}}$ in \cite{cannon2020unbalanced} and when ${3.3353\lambda\Delta_L\Delta_R < (1+\lambda)^{\delta_R/\Delta_L}}$ in~\cite{friedrich2020polymer}. Our result applies to a comparable parameter range: inequality~\eqref{eqn:unbalancedcond} holds, for instance, when 
$(1+e)\lambda\Delta_L\Delta_R < (1+\lambda)^{\delta_R/\Delta_L}$, or when $3\lambda\Delta_L\Delta_R < (1+\lambda)^{\delta_R/\Delta_L}$ and $\Delta_L < 6$.
More importantly, our algorithm is the first to achieve perfect sampling and near-linear running time.

\medskip\noindent\textbf{Potts model on expander graphs.} \ The $Q$-color ferromagnetic Potts model on a graph $G=(V,E)$ is the probability distribution $\mu_G^{\mathrm{Potts}}$ on the set of $Q$-colorings of the vertices of $G$; i.e., $\{1,\dots,Q\}^V$. Each $Q$-coloring $\sigma$ is assigned probability $\mu_G ^{\mathrm{Potts}}(\sigma) \propto {e^{\beta m(G,\sigma)}}$, where $m(G,\sigma)$ is the number of monochromatic edges of $G$ under the coloring $\sigma$ and $\beta > 0$ is a model parameter. 
When the parameter $\beta$ is large, and $G$ has some structure (e.g., $G$ is an expander graph), typical configurations drawn from $\mu_G ^{\mathrm{Potts}}$ are dominated by one of the $Q$ colors; that is, there is phase coexistence in the model.  This enables sampling using subset polymer models. 

Recall that an $n$-vertex graph $G = (V, E)$  is  an $\alpha$-expander if for all subsets $S \subseteq V$ with $|S| \leq n/2$, the number of edges in $E$ with exactly one endpoint in $S$ is at least $\alpha |S|$.

\begin{corollary}
	\label{corPottsEx}
	Consider the $Q$-color ferromagnetic Potts model on an $\alpha$-expander $n$-vertex graph of maximum degree $\Delta$. Suppose
	\begin{equation}
		\label{eqn:pottsbeta}
		\beta \geq \frac{1 + \log\left(\frac{\Delta+1}{e\Delta}+1\right) + \log((Q-1)\Delta )}{\alpha} \,.
	\end{equation} 
	Then there is a sampling algorithm with expected running time $O(n \log n)$ that outputs a sample $\sigma$ with distribution $\hat \mu$ so that $\|\hat \mu - \mu_G ^{\mathrm{Potts}}  \|_{\textsc{tv}} \le e^{-\Omega(n)}$.   
\end{corollary}

Previously,~\cite{chen2021fast} provided a $\varepsilon$-approximate sample for $\mu_G ^{\mathrm{Potts}}$ in $O(n \log(n/\varepsilon) \log(1/\varepsilon))$ time whenever $\beta \geq \frac{5+3\log((Q-1)\Delta)}{\alpha}$. 
Condition~\eqref{eqn:pottsbeta} holds when $\beta \geq \frac{1.2+\log((Q-1)\Delta)}{\alpha}$,
so our algorithm applies to a larger range of parameters and 
removes the dependence on $\varepsilon$ from the running time. We do not achieve perfect sampling in this application only because the subset polymer models used give approximations of $\mu ^{\mathrm{Potts}}_G$ rather than describing $\mu ^{\mathrm{Potts}}_G$ exactly. 

\subsection{Sampling unrooted graphlets in finite graphs}
\label{secUnRootedIntro}

As another application of our algorithm for sampling weighted rooted graphlets, 	
we consider next the problem of sampling weighted \emph{unrooted} graphlets in a finite graph.  Given a finite graph $G$, let $\cS(G)$ be the set of all graphlets of $G$.  Define the distribution $\nu_{G,\lam}$ on $\mathcal S(G)$ by
\[  \nu_{G,\lam} (\gamma) = \frac{ \lam^{|\gamma|}  }{ Z_{G,\lam}     }\,,  \quad \text{where}~Z_{G,\lam} = \sum\nolimits_{\gamma \in \cS(G)} \lam^{|\gamma|}  \,. \]
Read-McFarland and \v{S}tefankovi\v{c}~\cite{read2021hardness} gave a polynomial-time approximate sampling algorithm for $\nu_{G,\lam}$ for the class of maximum-degree $\Delta$ graphs when $\lam < \lam_*(\Delta)$ and prove that there is no such algorithm for $\lam \in (\lam_*(\Delta),1)$ unless NP=RP\footnote{In~\cite{read2021hardness}, the threshold is incorrectly stated as $\lam < \lam_*(\Delta+1)$; this is due to a minor error interchanging the infinite $\Delta$-regular tree with the infinite $\Delta$-ary tree; with this small correction their analysis goes through with the bound on $\lambda$ as stated here.}.
We give a new algorithm for this problem, covering the entire $\lambda < \lambda_*(\Delta)$ regime, and improving on the result of~\cite{read2021hardness} in two ways: (i) our running time is constant in expectation (with no dependence on $n$), while the running time of the $\varepsilon$-approximate sampler in~\cite{read2021hardness} is 
$n \cdot (1/\varepsilon)^{O(\log \Delta)}$;
and (ii) our algorithm outputs a perfect sample instead of an approximate one (and thus the running time has no dependence on any approximation parameter).

\begin{theorem}
	\label{thmUnrooted}
	Fix $\Delta \ge 3$ and let $\lam < \lam_*(\Delta)$.  Then for the class of finite graphs of maximum degree $\Delta$ there is a randomized algorithm running in constant expected time that outputs a perfect sample from $\nu_{G,\lam}$.  The expected running time is bounded as a function of $\Delta$ and $\lam$. 
\end{theorem}
\noindent
The algorithm we use for this theorem is a modification of the one for sampling rooted graphlets. We pick a uniformly random $v \in V$, run the same BFS percolation exploration, and accept the connected component of $v$ with an adjusted probability (to account for the fact that a graphlet can be generated from any of its vertices).
The acceptance probability is bounded away from $0$ and so the algorithm runs in constant expected time. As mentioned earlier, the $\varepsilon$-approximate sampling algorithm from~\cite{read2021hardness} is based on the exhaustive enumeration of all subgraphs of size $\le k$, for some $k$ that depends on $\varepsilon$. Our new algorithm entirely bypasses this enumeration barrier.

\section{Graphlet sampling: algorithms}
\label{secSubgraphs}

In this section we present our efficient perfect sampling algorithm
for weighted, vertex-labeled graphlets containing a fixed vertex $r$ from a maximum degree $\Delta$ graph;
in particular, in Section~\ref{subsec:rooted}, we prove a generalized version of Theorem~\ref{thmSubgraphSimple} from the introduction.
We also provide in Section~\ref{subsec:unrooted}
our algorithm for sampling weighted graphlets (i.e., the unrooted, unlabeled case) and establish Theorem~\ref{thmUnrooted}.
Our hardness results, that is~Lemmas~\ref{lemma:comp:hard} and \ref{lemma:part:finite}, are proved later in Section~\ref{sec:hardness}.

\subsection{Sampling rooted vertex-labeled graphlets}
\label{subsec:rooted}

Let $G = (V,E)$ be a (possibly infinite) graph of maximum degree $\Delta$. 
For $U \subseteq V$,
let $G[U]$ denote the corresponding vertex-induced subgraph of $G$; specifically, $G[U] = (U,E(U))$, where $E(U) \subseteq E$ is the set of edges of $G$ with both endpoints in $U$.
A vertex-induced subgraph is a \emph{graphlet} if it is connected.
For $r \in V$, let 	$\mathcal S(G,r)$ be the set of all graphlets of $G$ that contain vertex $r$.
We call the graphlets in $\mathcal S(G,r)$ the graphlets rooted at $r$.

Let $\Sigma = \{1,\dots,q\}$ be a set of vertex labels or colors, and
let $\mathcal S (G,r,q) = \bigcup_{(S,E(S)) \in \mathcal S(G,r)} \Sigma^S$ be the set of all vertex-labeled graphlets rooted at $r$.
Given a real parameter $\lambda > 0$,
we assign to each rooted vertex-labeled graphlet $\gamma \in \mathcal S (G,r,q) \cup \{\emptyset\}$ with $|\gamma|$ vertices the weight $w_\gamma = \lambda^{|\gamma|} f(\gamma)$, where $f:\mathcal S (G,r,q)\cup \{\emptyset\} \to [0,1]$. Note that $0 \leq w_\gamma \leq \lambda^{|\gamma|}$, which will be important for later~analysis.  

Define the probability distribution $\nu_{G,r,\lambda}$ on $\mathcal S (G,r,q) \cup \{\emptyset\}$ by setting 
\begin{equation}
	\nu_{G,r,\lambda}(\gamma) = \frac{w_\gamma}{Z(G,r,\lambda)}, 
\end{equation}
where $Z(G,r,\lambda) = \sum_{\gamma' \in \mathcal S(G,r,q) \cup \{\emptyset\}}w_{\gamma'}$. We assume that $G$, $f$, $q$ and $\lam$ are such that $Z(G,r,\lambda)$ is finite, so that this distribution is well defined. When $q=1$ and $f(\gamma) = \mathbbm{1}(\gamma \neq \emptyset)$, $\nu_{G,r,\lambda}$ corresponds exactly to the distribution defined in~\eqref{eq:rooted:dist} over the unlabeled graphlets of $G$ rooted at $r$.

We consider the problem of sampling from $\nu_{G,r,\lambda}$; this more general version of the sampling problem is later used as a subroutine for sampling polymer systems in Section~\ref{secPolymers}. 
Let
$$
\lambda_*(\Delta,q) := \frac{(\Delta-2)^{\Delta-2}}{q(\Delta-1)^{\Delta-1}};
$$
cf.,~\eqref{eqCrtiThresh}. Our main algorithmic result for sampling colored rooted graphlets is the following.

\begin{theorem}
	\label{thmSubgraphSample:lambda}
	Fix $\Delta \ge 3$, $q \ge 1$, and  suppose $0 < \lambda < \lambda_*(\Delta,q)$.  
	There is a randomized algorithm to exactly sample from $\nu_{G,r,\lam}$
	for graphs $G$ of maximum degree $\Delta$
	and functions $f:\mathcal S (G,r,q)\cup \{\emptyset\}  \to [0,1]$ where $f(\gamma)$ is computable in time polynomial in $|\gamma|$; this randomized algorithm has expected running time $O\big( Z_{G,r,\lam}^{-1}\big)$.
\end{theorem}

Theorem~\ref{thmSubgraphSimple} from the introduction corresponds to the special case when $q=1$ and $f(\gamma) = \mathbbm{1}(\gamma \neq \emptyset)$ (in this case $Z_{G,r,\lam} \ge \lambda$).
Other mild assumptions on the function $f$, e.g., $f(\emptyset)=1$ or $f(r) = 1$, ensure that $Z_{G,r,\lam}$ is bounded away from $0$ and, consequently, that the sampling algorithm in the theorem also has constant expected running time in those cases.

As a warm-up, let us consider first our algorithm for sampling labeled rooted graphlets on a finite graph $G=(V,E)$ with $f=1$, and purposely omit certain non-essential implementation details for clarity.
First, we find $p \in (0,1)$ such that $\frac{p}{q}(1-p)^{\Delta-2} = \lambda$; this choice of $p$ will be justified in what follows.
The algorithm then repeats the following process until a vertex-labeled graphlet is accepted:
\begin{enumerate}[topsep = 1pt]
	\setlength{\itemsep}{0pt}
	\item Each vertex of the graph is independently assigned with probability $p$ a uniform random color from $\{1,\dots,q\}$, or it is marked as ``not colored'' with the probability $1-p$.
	\item Let $\tilde \gamma$ be the vertex-labeled graphlet from $\mathcal S(G,r,q) \cup \{\emptyset\}$ corresponding the colored connected component of $r$; i.e., the set of vertices connected to $r$ by at least one path of colored vertices.
	\item Observe that the probability that $\tilde \gamma = \gamma$ is $({p}/{q})^{|\gamma|}(1-p)^{|\partial \gamma|}$, where 
	$\partial \gamma$ denotes to set of vertices in $G$ that are not in $\gamma$ but are adjacent to a vertex in $\gamma$
	(with a slight abuse of notation, we let $|\partial \emptyset|=1$).
	When $\tilde \gamma = \gamma$, our aim is to output $\gamma$ with probability $\propto \lambda^{|\gamma|}$ which has no dependence on $\partial \gamma$. Therefore, we use a ``rejection filter'' and only accept $\gamma$ with probability $(1- p)^{(\Delta-2)|\gamma| + 2 - |\partial \gamma|},$
	so that the probability that $\gamma$ is the output becomes:
	\begin{equation}
		\label{eq:rj:prob}
		\Big(\frac{p}{q}\Big)^{|\gamma|}(1-p)^{|\partial \gamma|}(1- p)^{(\Delta-2)|\gamma| + 2 - |\partial \gamma|} = (1-p)^2\Big(\frac{p}{q}(1- p)^{\Delta-2}\Big)^{|\gamma|} = (1-p)^2 \lambda^{|\gamma|}.
	\end{equation}
\end{enumerate}

From~\eqref{eq:rj:prob}, the choice of $p$ such that $\frac{p}{q}(1-p)^{\Delta-2} = \lambda$ is apparent. We will prove that only when $\lambda < \lambda_*(\Delta,q)$ there exists $p \in (0,1)$ such that $\frac{p}{q}(1-p)^{\Delta-2} = \lambda$.
In the actual implementation of the algorithm, it will in fact suffice to find an approximation for $p$.

We comment briefly on the intuition for the rejection filter. 
The acceptance probability must include a factor of $(1- p)^{- |\partial \gamma|}$, so that the final acceptance probability depends on $|\gamma|$ but not on $|\partial \gamma|$.
However, $(1- p)^{- |\partial \gamma|} > 1$ is not a valid probability,
so we use instead $(1- p)^{(\Delta-2)|\gamma| + 2 - |\partial \gamma|}$, which is at most $1$ since $(\Delta-2)|\gamma| + 2 \ge |\partial \gamma|$. 
This bound on $|\partial \gamma|$ is best possible since it is tight for the $\Delta$-regular tree. We note that using looser bounds for $|\partial \gamma|$ affects the range of the parameter $\lambda$ for which we can find $p \in (0,1)$ so that $\frac{p}{q}(1-p)^{\Delta-2} = \lambda$.

Finally, we mention that the algorithm as described requires $\Omega(|V|)$ time per iteration and cannot be extended to infinite graphs. This is easily corrected by assigning colors starting from $r$ and revealing only the colored component of $r$ in a breadth-first fashion. The threshold $\lambda_*(\Delta,q)$ is sharp in the sense that only when $\lambda < \lambda_*(\Delta,q)$ is the value of $p$ such that the revealing process is a sub-critical process that creates a small component with high probability.  This ensures the algorithm can be implemented efficiently. In particular, we stress that our algorithm avoids exhaustively enumerating labeled graphlets, as done in previous methods~\cite{chen2021fast}.


Before giving the implementation details of our algorithm and proving Theorem~\ref{thmSubgraphSample:lambda}, we consider the problem of finding $p \in (0,1)$ such that $\frac{p}{q}(1-p)^{\Delta-2} = \lambda$.
For $\Delta \ge 3$ and $q \ge 1$, consider the real function $g(x) = \frac{x}{q}(1-x)^{\Delta-2}$.
It can be readily checked that the function $g$ is continuous and differentiable in $[0,1]$, has a unique maximum at $x = \frac{1}{\Delta-1}$ with $g(\frac{1}{\Delta-1}) = \lambda_*(\Delta,q)$, is increasing in $[0,\frac{1}{\Delta-1}]$, and decreasing in $[\frac{1}{\Delta-1},1]$.
This implies that only when $\lambda < \lambda_*(\Delta,q)$,
there exists a value of $p \in [0,\frac{1}{\Delta-1})$ such that $g(p) = \lambda$.
In particular, when $\lambda > \lambda_*(\Delta,q)$, there is no value of $p$ for which $g(p) = \lambda$ and when $\lambda = \lambda_*(\Delta,q)$, the only possible value is $p = \frac{1}{\Delta-1}$. The latter case would result in a \emph{critical} percolation process, 
corresponding to the fact that the expected size of a graphlet from $\nu_{G,r,\lambda}$ has no uniform upper bound in the class of graphs of maximum degree $\Delta$; in fact, it is infinite on the $\Delta$-regular tree.
We can find a suitable approximation for $p$ when $\lambda < \lambda_*(\Delta,q)$ via a simple (binary search) procedure.

\begin{lemma}
	\label{lemma:binary:search}
	For any $\lambda \in [0,\lambda_*(\Delta,q))$ we can find rational numbers $\hat p \in [0,\frac{1}{\Delta-1})$ and $\hat\lambda \in [\lambda,\lambda_*(\Delta,q)]$ 
	such that $g(\hat p) = \hat \lambda$ in $O(|\log \frac{1}{\Delta q(\lambda_* - \lambda)}|)$ time. 
\end{lemma}


\noindent The proof of this lemma 
appears after the proof of Theorem~\ref{thmSubgraphSample:lambda}.
We now prove Theorem~\ref{thmSubgraphSample:lambda}, including giving a more detailed version of the algorithm outlined above that includes the previously omitted implementation details and allows for general functions $f:\mathcal S (G,r,q)\cup \{\emptyset\}  \to [0,1]$. 

\begin{proof}[Proof of~Theorem~\ref{thmSubgraphSample:lambda}]
	For ease of notation, let $\lambda_* = \lambda_*(\Delta,q)$.
	Our algorithm to sample from $\nu_{G,r,\lam}$ when $\lambda < \lambda_*$ explores from $r$ in a breadth-first manner and stops once it has revealed the colored connected component of $r$. It proceeds as follows:
	
	\begin{enumerate}[topsep = 0pt]
		\setlength{\itemsep}{0pt}
		\item 
		Find $\hat p \in [0,\frac{1}{\Delta-1})$ and $\hat \lambda \in [\lambda,\lambda_*)$ such that $g(\hat p) = \hat \lambda$.
		This can be done in $O(|\log \frac{1}{\Delta q(\lambda_* - \lambda)}|)$ time per Lemma~\ref{lemma:binary:search}.
		
		
		\item Let $Q$ be a queue. With probability $1-\hat p$ do not add $r$ to $Q$; otherwise, assign $r$ a color uniformly at random from $\{1,\dots,q\}$ and add $r$ to $Q$. Mark $r$ as explored. 		
		
		\item While $Q \neq \emptyset$, repeat the following:
		
		\begin{enumerate}[topsep = 0pt]
			\setlength{\itemsep}{0pt}
			\item[3.1)] Pop a vertex $v$ from $Q$.
			\item[3.2)] For each unexplored neighbor $w$ of $v$, with probability $1-\hat p$ do not add $w$ to $Q$; otherwise, assign $w$ a color uniformly at random from $\{1,\dots,q\}$ and add $w$ to $Q$. Mark $w$ as explored (regardless of whether it was added to $Q$ or not).
		\end{enumerate}
		
		\item Let $\gamma$ be the vertex-labeled graphlet from $\mathcal S(G,r,q) \cup \{\emptyset\}$ corresponding the colored connected component of $r$. Accept $\gamma$ with probability:
		\[
		f(\gamma) \cdot (1-\hat p)^{(\Delta-2)|\gamma| + 2 - |\partial \gamma|} \Big(\frac{\lambda}{\hat\lambda}\Big)^{|\gamma|}.
		\]	
		\item If $\gamma$ is rejected, go to Step 2 and repeat.
	\end{enumerate}
	
	\noindent
	The probability of obtaining $\gamma \in \mathcal S(G,r,q) \cup \{\emptyset\}$ in an iteration of the algorithm is:
	$$
	\Big(\frac{{\hat p}}{q}\Big)^{|\gamma|} (1-\hat p)^{|\partial \gamma|}\cdot f(\gamma) (1-\hat p)^{(\Delta-2)|\gamma| + 2 - |\partial \gamma|} \Big(\frac{\lambda}{\hat\lambda}\Big)^{|\gamma|}= (1-\hat p)^2 f(\gamma) \,\lambda^{|\gamma|} = (1-\hat p)^2 w_\gamma,
	$$	
	and thus the overall acceptance probability in an iteration is:
	$$
	\rho := (1-\hat p)^2 \sum_{\gamma \in \mathcal S(G,r,q)\cup\{\emptyset\} } w_\gamma = (1-\hat p)^2 Z_{G,r,\lambda} .
	$$
	Then,
	$$
	\Pr[\gamma \in \mathcal S(G,r,q)\cup \{\emptyset\}~\text{is the output}] = \sum\nolimits_{t \ge 1} (1-\hat p)^2 w_\gamma (1-\rho)^{t-1} = \frac{(1-\hat p)^2 w_\gamma}{\rho} = \nu_{G,r,\lam}(\gamma).
	$$
	
	We next bound the expected running time of the algorithm.
	We claim first that expected running per iteration is at most a constant that depends only on $a$, $\Delta$ and $q$. 
	If $\gamma$ is the configuration generated in an iteration, 
	it is discovered in $O(|\gamma| + |\partial\gamma|) = O(|\gamma|)$ time and,
	by assumption, $f(\gamma)$ can be computed in at most $O(|\gamma|^a)$ time, for suitable a constant $a > 0$.
	Let $\hat\mu$ the output distribution of Step 3 of the algorithm.
	Then, there exists a constant $C = C(q,\Delta) > 0$ such that the expected running time of each iteration is at most:
	\begin{align}
		\label{eq:alg:1}
		C \sum_{\gamma \in \mathcal S(G,r,q) \cup\emptyset} |\gamma|^{\max\{a,1\}} \, {\Pr}_{\hat\mu}[\gamma] = C\cdot \E_{\hat\mu}[|\gamma|^{\max\{1,a\}}].
	\end{align}

	We show next that $|\gamma|$ (under $\hat\mu$) is stochastically dominated by 
	a random variable $W = X+Y$ (i.e., $|\gamma| \prec W$), 
	where $X$ and $Y$
	are i.i.d.\ random variables corresponding to the cluster size of a homogeneous Galton-Watson process with offspring distribution $Bin(\Delta-1,\hat p)$.
	We recall that this is the branching process that starting from a single vertex (or individual) $N_0$, adds $Z_1 \sim Bin(\Delta-1,\hat p)$ descendants to $N_0$. The process is then  repeated for each new descendant. The process can either die out or go on forever; the cluster size is the number of descendants of $N_0$.
	When different offspring of $N_0$ use different distributions to generate its descendants, the process is called heterogeneous (see, e.g., \cite{RandomGraphBook} for additional background). 
	
	To see that $|\gamma| \prec W = X+Y$, first note that $|\gamma| \prec L$, where 
	$L$ is the cluster size of a heterogeneous Galton-Watson process, in which the root vertex has offspring distribution $Bin(\Delta,\hat p)$ and
	every other vertex has offspring distribution $Bin(\Delta-1,\hat p)$.
	This is because the branching process generating $\gamma$ includes the root only with probability $\hat p$ (the root is always present in the Galton-Watson process), and, in addition, it considers at most $\Delta$ (from the root) or $\Delta-1$ (from any other vertex) potential branches (or descendants).	
	In turn, we can bound the cluster size $L$ by $L \prec X+Y$, since,
	in the branching process corresponding to $L$,
	we can couple the first $\Delta-1$ branches of the root $N_0$ with $X$ (starting at the root) and the remaining branch with $Y$ (starting at the child of the root not coupled with $X$).
	
	It is well-known that $X$ and $Y$ have finite moments
	when $(\Delta-1)\hat p < 1$ (see, e.g.,~\cite{RandomGraphBook}). In particular, there exists a constant $A = A(a,\Delta,\hat p) > 0$ such that
	\begin{equation}
		\label{eq:alg:2}
		\E_{\hat\mu}[|\gamma|^a] \le \E[L^a] \le \E[(X+Y)^a] \le 2^{a}(\E[X^a]+\E[Y^a]) \le A.
	\end{equation}
	This together with~\eqref{eq:alg:1} shows that the expected running time in each iteration of the algorithm is bounded by $C \cdot A$.
	
	Now, let 
	$R$ be the number of times Steps 2--5 are repeated, let
	$T$ be the overall running time of the algorithm. Then:
	\begin{align}
		\label{eq:root:exp-rt}
		\E[T] &= \sum_{t\ge 1} \E[T \mid R=t] \Pr[R = t] \le C \cdot A \cdot \sum_{t\ge 1} t (1-\rho)^{t-1} \rho \le \frac{CA}{\rho}, 
	\end{align}
	and the result follows.	
\end{proof}

We conclude this section with the proof of Lemma~\ref{lemma:binary:search}.

\begin{proof}[Proof of Lemma~\ref{lemma:binary:search}]
	It suffices to find $\hat p \in [p,\frac{1}{\Delta-1}]$. This can be done via binary search in $t$ steps, provided $t \ge 0$ is such that
	$
	\frac{1}{\Delta-1} \cdot \frac{1}{2^t} \le \frac{1}{\Delta-1}-p.
	$
	Since $g' \le \frac{1}{q}$, it follows from the mean value theorem that
	$
	{q (\lambda_* - \lambda)} \leq \frac{1}{\Delta-1} - p .
	$
	Thus for the binary search to require at most $t$ steps it is sufficient to pick $t$ so that
	$
	\frac{1}{\Delta-1} \cdot \frac{1}{2^t} \le {q(\lambda_*-\lambda)},
	$
	and the result follows.	
\end{proof}

\subsection{Sampling unrooted graphlets}
\label{subsec:unrooted}

We consider next the problem of sampling weighted graphlets from a {finite} graph $G = (V,E)$ of maximum degree $\Delta$; specifically, in this variant of the sampling problem we consider unrooted, unlabeled, weighted graphlets of $G$.
Let $\mathcal S(G)$ be the set of all graphlets of $G$.
We define the probability distribution $\nu_{G,\lambda}$ on $\mathcal S(G)$ by setting 
\begin{equation*}
	\nu_{G,\lambda}(S) = \frac{\lambda^{|S|}}{Z_{G,\lambda}}, 
\end{equation*}
where $Z_{G,\lambda} = \sum_{S' \in \mathcal S(G) \cup \{\emptyset\}}\lambda^{|S'|}$.
The problem of (approximately) sampling from $\nu_{G,\lambda}$ 
is quite natural. 
In~\cite{read2021hardness}, it was established that this problem is computationally hard when $\lambda > \lambda_*(\Delta) = \frac{(\Delta-2)^{\Delta-2}}{(\Delta-1)^{\Delta-1}}$;
an $\varepsilon$-approximate sampling algorithm was also given in~\cite{read2021hardness} for the case when $\lambda < \lambda_*(\Delta)$
with running time $n \cdot (1/\varepsilon)^{O(\log \Delta)}$. 
We establish the following:
\begin{theorem}
	\label{thmSubgraphSample:lambda:unrooted}
	Suppose $\Delta \ge 3$ and $\lam > 0$ are such that $\lambda < \lambda_*(\Delta)$. 
	There is a randomized algorithm to exactly sample from $\nu_{G,\lam}$ with $O(1)$ expected running time for finite graphs $G$ of maximum degree~$\Delta$.
\end{theorem}

\begin{proof}
	For ease of notation, we set $\lambda_* = \lambda_*(\Delta)$ throughout this proof.
	Our algorithm to sample from $\nu_{G,\lambda}$ is 
	based on the algorithm to sample from $\nu_{G,r,\lambda}$ (the rooted, vertex-labeled, weighted case).
	The idea is to pick a root uniformly at random and run the algorithm 
	for the rooted case from this random vertex with the rejection filter adjusted to account for the fact that a graphlet can be generated from any of its vertices.
	It proceeds as follows:
	\begin{enumerate}
		\item 
		Find $\hat p \in [0,\frac{1}{\Delta-1})$ and $\hat \lambda \in [\lambda,\lambda_*)$ such that $g(\hat p) = \hat \lambda$ using the method from Lemma~\ref{lemma:binary:search}.
		
		\item Pick a vertex $r \in V$ uniformly at random. 
		
		
		\item Let $Q$ be a queue. With probability $\hat p$ add $r$ to $Q$ and mark it as colored. Mark $r$ as explored. 		
		
		\item While $Q \neq \emptyset$, repeat the following:
		
		\begin{enumerate}
			\item[3.1)] Pop a vertex $v$ from $Q$.
			\item[3.2)] For each unexplored neighbor $w$ of $v$, with probability $\hat p$ add $w$ to $Q$ and mark $w$ as colored. Mark $w$ as explored.
		\end{enumerate}
		
		\item Let $S\in\mathcal S(G)$ be the graphlet corresponding to the colored connected component of $v$. Accept $S$ with probability:
		$$
		\frac{1}{|S|} \cdot (1-\hat p)^{(\Delta-2)|S| + 2 - |\partial S|} \Big(\frac{\lambda}{\hat\lambda}\Big)^{|S|}.
		$$	
		\item If $S$ is rejected, go back to Step 2 and repeat.
	\end{enumerate}
	
	The analysis of this algorithm is similar to that in the proof of Theorem~\ref{thmSubgraphSample:lambda}.
	Let $n = |V|$. The probability that the algorithm outputs $S$
	in an iteration is:
	\begin{equation}
		\label{eq:unr:bound}
		\sum_{v \in S} \frac{1}{n} \cdot
		{\hat p}^{|S|} (1-\hat p)^{|\partial S|}\cdot \frac{1}{|S|} \cdot (1-\hat p)^{(\Delta-2)|S| + 2 - |\partial S|} \Big(\frac{\lambda}{\hat\lambda}\Big)^{|S|} = \frac{(1-\hat p)^2 \lambda^{|S|}}{n}.
	\end{equation}	
	Hence, conditioned on acceptance, the probability of obtaining $S \in \mathcal S(G)$ is thus $\nu_{G,\lam}(S)$, and so the output distribution of the algorithm is $\nu_{G,\lam}$.
	
	For the running time of the algorithm, we note that Step 4 of the algorithm is analogous to Step 3 of the algorithm in the proof of Theorem~\ref{thmSubgraphSample:lambda}, and so the expected running time of each round is
	also bounded by a constant $C = C(\Delta,\hat p) > 0$.	
	Let $T$ be the overall the running time of the algorithm.
	From~\eqref{eq:unr:bound}, we have that the overall acceptance probability in a round is $\rho = \frac{(1-\hat p)^2 Z(G,\lambda)}{n}$. Then,
	as in~\eqref{eq:root:exp-rt}, we deduce that $\E[T] = O(n Z(G,\lambda)^{-1})$.
	Since $Z(G,\lambda) \ge n \lambda$, we have $\E[T] = O(1)$.	 
\end{proof}

\section{Applications to Polymer Models}
\label{secPolymers}

In this section, we show how to use our algorithm for sampling rooted vertex-labeled graphlets from Section~\ref{secSubgraphs}  to sample from subset polymer models and prove Theorem~\ref{thmPolymerSample1}. 

Consider a subset polymer model on an $n$-vertex graph $G=(V,E)$; see Section~\ref{secPolymerIntro} for the definition. Recall that we use $\cC_v$ for the set of all polymers containing vertex $v \in V$, and
let $\gamma_\emptyset$ denote the empty polymer. Define the distribution $\nu_v$ on $\cC_v \cup \{\gamma_\emptyset\}$ where 
$
\nu_v ( \gamma) = \frac{w_{\gamma}}{40}
$
and 
$$
\nu_v ( \gamma_\emptyset) = 1 - \frac{1}{40 } \sum_{\hat\gamma \in \cC_v } w_{\hat\gamma}.
$$
\noindent
We note that $\nu_v$ is well-defined when condition \eqref{eqPolyLamCondition} holds, since under this condition we have
$\sum_{\hat\gamma \in \cC_v } w_{\hat\gamma} \le 40$; this is proved later in Lemma~\ref{lemma:part:finite-general}.
The following Markov chain on $\Omega(\cC)$ is similar to the one introduced in~\cite{chen2021fast}.

\bigskip\noindent
\textbf{Polymer dynamics.}
Given a configuration $X_t \in \Omega(\cC)$, form $X_{t+1}$ as follows:
\begin{enumerate}
	\item Pick $v \in V$ uniformly at random and let $S_v = \{ \gamma \in X_t: v \in \gamma\}$ (note that $S_v$ is either empty or contains $1$ polymer).
	\item With probability $1/41$,   let $X_{t+1} = X_t \setminus S_v$. 
	\item Otherwise, with the remaining probability $40/41$, sample $\gamma$ from $\nu_v$.  Let $X_{t+1} = X_t \cup \{\gamma\}$ if $X_t \cup \{\gamma\} \in \Omega(\cC)$ and let $X_{t+1} = X_t$ otherwise. 
\end{enumerate}

We first note that this Markov chain is irreducible (every configuration can reach and can be reached from the empty set of polymers with positive probability), and
aperiodic (there is positive probability of remaining on the same state). It is also reversible with respect to $\mu$: letting $P$ be the transition matrix for this chain, $\Gamma \in \Omega(\cC)$, and a polymer $\gamma \not\in \Gamma$ be such that $\Gamma \cup \{\gamma\} \in \Omega(\cC)$, we have
$
{\mu(\Gamma \cup \{\gamma\})}/{\mu(\Gamma)} = w_{\gamma} 
$
and for $\gamma \neq \gamma_\emptyset$:
$$
\frac{P(\Gamma,\Gamma \cup \{\gamma\})}{P(\Gamma \cup \{\gamma\},\Gamma)} = \frac{\sum_{v \in \gamma} \frac{40}{41n} \nu_v(\gamma)}{\frac{|\gamma|}{41n}} = w_{\gamma}.
$$

Now, to implement a single update step of the polymer dynamics, one must sample from $\nu_v$ in Step 3. 
We give a fast perfect sampler for $\nu_v$.

\begin{theorem}
	\label{thmPolymerSample}
	Fix $\Delta \ge 3, q\ge 1$, $\delta > 0$ and let $\lambda < \lambda_*(\Delta,q)$.
	Consider a family of $\delta$-covering subset polymer models defined on graphs of maximum degree $\Delta$ with computationally feasible weights that satisfy $w_\gamma \leq \lambda^{|\gamma|}$.
	There is a randomized  algorithm to sample perfectly from $\nu_v$ for any $v \in V(G)$ with a constant expected running time.
\end{theorem}
\begin{proof}
	For each vertex $v$,
	let $\hat\nu_v$ be the distribution over $\cC_v \cup \{\gamma_\emptyset\}$ where 
	$\hat \nu_v ( \gamma) = {w_{\gamma}}/{Z_v}$
	with $Z_v = \sum _{\hat\gamma \in \cC_v \cup \{\gamma_\emptyset\} } w_{\hat\gamma}$
	and $w_{\gamma_\emptyset} = 1$. By Theorem~\ref{thmSubgraphSample:lambda}, we can perfectly from $\hat \nu_v$ in constant expected time, with the constant depending $\lam, \Delta$, and $q$.
	Our algorithm to sample from $\nu_v$ draws a perfect sample from $\hat \nu_v$ first and then outputs the sample with probability $Z_v/40$; otherwise it outputs the empty polymer $\gamma_\emptyset$.
	The output distribution of this algorithm is exactly $\nu_v$, and so all that remains is for us to show how to sample exactly from a Bernoulli distribution with parameter $Z_v/40$, denoted $\mathrm{Ber}(Z_v/40)$.
	
	We can sample exactly from $\mathrm{Ber}(1/Z_v)$ in constant expected time by drawing an exact sample from $\hat \nu_v$ (using the algorithm from Theorem~\ref{thmSubgraphSample:lambda}) and outputting $1$ if the sample is $\gamma_\emptyset$ and $0$ otherwise.
	With this, we can then use a ``Bernoulli factory'' to obtain a perfect sample from
	$\mathrm{Ber}(Z_v/40)$. Specifically,
	we know how to simulate a Bernoulli coin with parameter $p=1/Z_v$ and require a Bernoulli coin with parameter $f(p)=1/(40p) =  Z_v/40$.
	This is possible in constant expected time 
	since $1+\delta \le Z_v \le 39$;
    the lower bound follows directly from the $\delta$-covering assumption and the upper is implied by~\eqref{eqPolyLamCondition} (as mentioned, this is proved in Lemma~\ref{lemma:part:finite-general}). 
    Thus, $p \in [\frac{1}{39},\frac{1}{1+\delta}] \subset (0,1)$, $f(p) \in [\frac{1+\delta}{40},\frac{39}{40}] \subset (0,1)$,
    and the function $f(p) = 1/(40p)$ is real analytic in this interval. Theorem 2 from~\cite{NP} then yields the desired Bernoulli sample in constant expected time.  
\end{proof}

Note that Bernoulli factories have been used in a similar fashion to design perfect sampling algorithms for CSP solutions and spin models in~\cite{he2022sampling,he2023improved,anand2023perfect}.

Using this theorem,
we give next a perfect sampling algorithm that works whenever a new condition~\eqref{eqPerfectSamplingCondition} is satisfied (our algorithm also requires the assumptions in Theorem~\ref{thmPolymerSample}). 

\subsection{Perfect Sampling for polymer systems: Proof of Theorem~\ref{thmPolymerSample1}} 

We propose here an algorithm to output a perfect sample from $\mu$. Our algorithm 
is based on the polymer dynamics and
the coupling from the past method~\cite{ProppWilson}, using the notion of bounding Markov chains \cite{huber2004bounding,HN} to efficiently implement it. 

We proceed with the proof of Theorem \ref{thmPolymerSample1}. We start with the description of a grand coupling for the polymer dynamics, 
which is then used to implement a coupling from the past algorithm.
For an $n$-vertex graph $G = (V,E)$, let $\{X_t^\Gamma\}$ denote an instance of the polymer dynamics started from 
the polymer configuration $\Gamma \in \Omega(\mathcal C)$.
For all $\Gamma \in \Omega(\mathcal C)$, the chains $\{X_t^\Gamma\}$ are
coupled by choosing the same uniform random random vertex $v \in V$, the same polymer $\gamma$ sampled from $\nu_v$, and the same uniform random number in $[0,1]$ to decide whether to remove $S_v$ (Step 2) or to add $\gamma$ (Step 3).
A coupling from the past algorithm will find a time $-T$ such the grand coupling 
started from all possible states at time $-T$ 
coalesces to a single state by time $0$.
This guarantees that the output of the algorithm, that is the state at time $0$,
has distribution $\mu$ (see~Theorem 1 from~\cite{ProppWilson}).
Such a $T$ can be found with a binary search procedure.  Unfortunately,
implementing the coupling from the past algorithm 
in this manner for the polymer dynamics Markov chain is infeasible, since it requires simulating an exponential number of copies of the polymer dynamics, one from each $\Gamma \in \Omega(\mathcal C)$.

To work around this, we consider a bounding Markov chain for the polymer dynamics rather than the polymer dynamics chain itself. Bounding Markov chains were pioneered in~\cite{huber2004bounding,HN}
as a method for efficiently implementing coupling from the past.
The bounding chain for the polymer dynamics has state space $\Omega(\mathcal C) \times 2^\mathcal C$ and will be denoted by $\{B_t,D_t\}$, where $B_t \in \Omega(\mathcal C)$ and $D_t \subseteq \mathcal C$ are sets of polymers. The chain will maintain throughout that all polymers in $B_t$ are compatible and that every polymer in $B_t$ is compatible with every polymer in $D_t$.  The polymers in $D_t$ do {\it not} need to be compatible with each other. A step of the bounding Markov chain is defined next.

\bigskip\noindent
{\bf Polymer Dynamics Bounding Chain.} Given $\{B_t,D_t\}$, the chain generates $\{B_{t+1},D_{t+1}\}$ by:
\begin{enumerate}
	\item Uniformly at random, select $v \in V$. \label{item:pickv}
	\item With probability $1/41$, remove all polymers containing $v$ by setting $B_{t+1} = B_t \setminus \mathcal{C}_v$ and $D_{t+1} = D_t \setminus \mathcal{C}_v$.  \label{item:delete} 
	\item With the remaining probability $40/41$, draw a sample $\gamma$ according to $\nu_v$ and: \label{item:comparetoB}
	\begin{enumerate}
		\item If $\gamma$ is compatible with $B_t$ and $\gamma$ is compatible with $D_t \setminus \{\gamma\}$, let $B_{t+1} = B_t \cup \{\gamma\}$ and let $D_{t+1} = D_t \setminus \{\gamma\}$. \label{item:addtoB}
		
		\item Else if $\gamma$ is compatible with $B_t$ but $\gamma$ is not compatible with $D_t$, let $B_{t+1} = B_t$ 
		and let $D_{t+1} = D_{t} \cup \{\gamma\}$.
		\label{item:addtoD}
		
		\item  If $\gamma$ is incompatible with $B_t$, do nothing: $B_{t+1} = B_t$ and $D_{t+1} = D_t$. 
		\label{item:incompatwB}
	\end{enumerate}
\end{enumerate}
\noindent
Observe that polymers are only added to $B_t$ if they are compatible with all other polymers in $B_t$; hence, if $B_0$ is a valid polymer configuration, so is $B_t$   for all $t \ge 0$.

To implement a step of the polymer dynamics bounding chain it suffices to
pick a vertex $v \in V$ uniformly at random,
a uniform random number in $[0,1]$,
and a polymer $\gamma$ from $\nu_v$, just like for the polymer dynamics. 
Hence, we can couple the evolution of $\{B_t,D_t\}$ with the grand coupling of the polymer dynamics described earlier. 
If we set $B_0 = \emptyset$ and $D_0 = \mathcal C$, it can be checked that for all $\Gamma \in \Omega(\mathcal C)$ and all $t \ge 0:$
$$B_t \subseteq X_t^\Gamma \subseteq B_t \cup D_t.$$
Indeed, this holds initially for $t = 0$, and the grand coupling ensures that whenever a polymer is removed from $X_t^\Gamma$ it is also removed from $B_t$, and whenever a polymer is added to $X_t^\Gamma$ it is also added to $B_t$ or $D_t$.
Consequently, $\{B_t,D_t\}$ is a bounding chain for the polymer dynamics.
In particular, the first time $B_t =  B_t \cup D_t$, all instances $X_t^\Gamma$ have necessarily coalesced to the same configuration.
This bounding chain allows us to implement the coupling from the past algorithm efficiently, as follows. 


\bigskip\noindent
{\bf Coupling from the Past.} Set $k=1$.
\begin{enumerate}[label =(\Alph*)]
	\item For $t=-2^{k},-2^{k}+1,\dots,-2^{k-1}$ generate $\rho_t = (v_t,\gamma_t,r_t)$ by choosing $v_t \in V$ uniformly at random, $r_t \in [0,1]$ uniformly at random, and by sampling $\gamma_t \in \mathcal C_{v_t} $ from $\nu_{v_t}$.
	\item Set $B_{-2^{k}}= \emptyset$ and $D_{-2^{k}} = \mathcal{C}$.
	\item Simulate the polymer dynamics bounding chain from time $-2^{k}$ to time $0$ using $\rho_{-2^{k}},\dots,\rho_{-1}$. 
	\item If $B_0 = B_0 \cup D_0$, then output $B_0$; otherwise set $k \rightarrow 2k$ and repeat the process from Step~(A).
\end{enumerate}

	
	
	

This implementation of the coupling from the past algorithm provides a perfect sample from $\mu$; see~\cite{ProppWilson}. It remains for us to show that it is efficient.
For this, we show first that the expected number of steps
of the polymer dynamics bounding chain throughout the execution of the algorithm is $O(n \log n)$. Afterwards, we will show how to implement steps so that they can be executed in amortized constant expected time.

\begin{lemma}\label{lemCFTP}
	Suppose a subset polymer model on an $n$-vertex graph satisfies condition \eqref{eqPerfectSamplingCondition}. Then,
	the expected number of steps of 
	the polymer dynamics bounding chain in
	the coupling from past algorithm is $O(n \log n)$.
\end{lemma}
\begin{proof}
	Let us first consider an evolution of the bounding chain $\{B_t,D_t\}$
	from the initial state $\{B_0,D_0\} = \{\emptyset,\mathcal C\}$. Let 
	$T$ be the first time $B_t = B_t \cup D_t$; that is, the time until $D_t = \emptyset$.
	We bound first the expected value of $T$. Let $\phi_t = |(\cup_{\gamma \in D_t} \gamma) \cap V|$, the number of vertices of $G$ that are included in at least one polymer of $D_t$. Note that $\phi_t \leq n$. We analyze the expected value of $\phi_{t+1} - \phi_t$. 
	
	With probability $\phi_t/(41n)$, a vertex in a polymer of $D_t$ is picked and all polymers in $D_t$ containing that vertex are removed, in which case $\phi_{t+1} \leq \phi_t - 1$. 
	(A polymer may also be removed from $D_t$ in Step 3(a), but this case is omitted from our calculations as it never increases $\phi_t$.)
	
	The potential $\phi_{t}$ may increase in Step 3(b), when a polymer $\gamma$ may be added to $D_t$. Such a polymer is only added if it is incompatible with $D_t$. The probability this occurs is at most:
	\begin{align*}
		\sum_{v \in D_t} \frac{40}{41n} \sum_{\gamma \not\sim {v} } \nu_v(\gamma)
	\end{align*}
	For each such polymer, the expected increase in the size of $\phi_t$ is at most $|\gamma|$.	We can calculate: 
	\begin{align*}
		\mathbb{E}[\phi_{t+1} - \phi_{t} \mid B_t,D_t] \leq -\frac{\phi_t}{41n}   + \sum_{v \in D_t} \frac{40}{41n} \sum_{\gamma \not\sim {v} } |\gamma| \nu_v(\gamma) = -\frac{\phi_t}{41n}   +  \frac{\phi_t}{41n} \sum_{\gamma \not\sim {v} } |\gamma| w_\gamma.
	\end{align*}
	By \eqref{eqPerfectSamplingCondition}, there exists a constant $\theta < 1$ such that $\sum_{\gamma \not\sim v } |\gamma| w_\gamma \leq \theta$. It follows that
	\begin{align*}
		\mathbb{E}[\phi_{t+1} - \phi_{t} \mid B_t,D_t] \leq -\frac{\phi_t}{41n} (1 - \theta). 
	\end{align*}
	This shows that $\{\phi_t\}$ is a stochastic process with
	variable multiplicative drift.
	Since $T$ is also the first time $\phi_t < 1$, using standard hitting time estimates (see, e.g., Theorem 10 in~\cite{KK-drift}), we get:
	$$
	\E[T] \le \frac{41n}{1-\theta} + \frac{41n}{1-\theta} \int_{1}^{n} \frac{1}{z} dz = O(n \log n).
	$$
	Now, let $-M$ be the first time in past such that, 
	after setting $B_{-M} = \emptyset$ and $D_{-M} = \mathcal C$, we have $B_{0} =B_0\cup D_0$. 
	The coupling from the past algorithm simulates at most $2M + M +  M/2 + \dots + 1 = O(M)$ steps of the bounding chain. The result follows by noting that $\E[M] = \E[T] = O(n \log n)$.
\end{proof}

It remains to consider how to efficiently implement the steps of the polymer dynamics bounding chain. This is subtle because $D_t$ may initially contain an exponentially large number of polymers, and care is thus needed in how $D_t$ is represented and stored. We describe and analyze efficient data structures for $B_t$ and $D_t$ in the following lemma.

\begin{lemma}
	\label{lemDataStructure}
	Suppose a subset polymer model on an $n$-vertex graph of maximum degree $\Delta \ge 3$ with computationally feasible weights that satisfy $w_\gamma \leq \lambda^{|\gamma|}$ for some $\lambda < \lambda_*(\Delta,q)$.
	There exists a compact representation of $B_t$ and $D_t$ that uses $O(n + t)$ space in expectation. Using this representation, each iteration of the Polymer Dynamics Bounding Chain can be executed in amortized constant expected time, and the termination condition $B_t = B_t \cup D_t$ can be checked in constant time. 
\end{lemma}

\begin{proof}
	First, by Theorem~\ref{thmPolymerSample}, there is a randomized algorithm to sample perfectly from $\nu_v$ for any $v \in V$ with constant expected running time. Therefore, the expected size of the polymer $\gamma$ produced by the sampling procedure is constant.  This will imply many of our data structure operations, which take (amortized) time $O(|\gamma|)$ for some $\gamma$ sampled from $\nu_v$, take constant expected time. 
	
	For simplicity, we assume time begins at $t = 0$ and at some point in the future we wish to check whether $B_t = B_t \cup D_t$. Shifting the indices appropriately, this data structure can be applied in each round of coupling from the past from each negative starting time. Note because all polymers in $B_t$ are compatible with all polymers in $D_t$, the condition $B_t = B_t \cup D_t$ is equivalent to the condition $D_t = \emptyset$.

	{\it Data Structure for $B_t$}: Note $B_t$ is a collection of compatible polymers. The operations to be performed on $B_t$ include addition of a polymer (Step \ref{item:addtoB}), deletion of any polymer containing a particular vertex (Step \ref{item:delete}), and comparison of a polymer to $B_t$ to determine whether it is compatible with $B_t$ (Step \ref{item:comparetoB}). $B_t$ will be stored as a length $n$ array, and we use $\overline{B}_t$ to denote the state of the array at time $t$. Initially, $\overline{B}_0$ has only $0$ entries. Each entry may also have a pointer to a node of a doubly linked list. Each polymer in $B_t$ is stored as a doubly linked list. This requires $O(n)$ space. 
	
	When $\gamma$ is added to $B_t$, for each $v \in \gamma$ we set $\overline{B}_t(v) = 1$. We also create a doubly linked list $\cL^\gamma$, which has a node for each $v \in \gamma$ that notes the color $v$ is assigned in $\gamma$ and contains a forward pointer, a backward pointer, and a pointer back to $\overline{B}_t(v)$. Each entry of $\overline{B}_t$ for $v \in \gamma$ also points to the corresponding node of $\cL^\gamma$. Doing these updates to the data structure for $B_t$ takes time $O(|\gamma|)$.

	The polymer removal step only occurs in Step \ref{item:delete} when any polymers containing $v$ are removed from $B_t$. There is at most one such polymer $\gamma$ in $B_t$ containing $v$. The other vertices in $\gamma$ can quickly be found using the doubly linked list $\mathcal{L}^\gamma$ and its pointers back to $\overline{B}_t$. We can set $B_t(w) = 0$ for each $w \in \gamma$ and remove $\mathcal{L}^\gamma$.  This removal of $\gamma$ from $B_t$ occurs in time $O(|\gamma|)$. 
	
	To compare a polymer $\gamma$ to $B_t$ to determine whether it is compatible with $B_t$, one looks up every vertex $w$ that is in or adjacent to  $\gamma$ in $\overline{B}_t$ to see if it is part of a polymer of $B_t$, that is, if $\overline{B}_t(w) = 1$. If at least one of these vertices $w$ has $\overline{B}_t(w) = 1$, $\gamma$ is incompatible with $B_t$; if all these  $w$ have $\overline{B}_t(w) = 0$, then $\gamma$ is compatible with $B_t$. This takes time $O(\Delta|\gamma|)$, which because we assume $\Delta$ is constant is $O(|\gamma|)$. 
	
	Note each of these operations can be performed in time $O(|\gamma|)$, where $\gamma$ is a polymer that was at some (possibly earlier) time step drawn from $\nu_v$ for some $v$. When each polymer $\gamma$ is sampled in Step~\ref{item:comparetoB} in some iteration of the Polymer Dynamics Bounding Chain, all $O(|\gamma|)$ operations that can potentially be performed on it in the future are immediately charged to it. As argued in the first paragraph above, the expected size of a polymer $\gamma$ drawn in Step~\ref{item:comparetoB} is constant, so the amount $O(|\gamma|)$ charged to each sampled polymer $\gamma$ is constant in expectation. We conclude updating $B_t$ takes expected amortized constant time. This data structure requires $O(n)$ space.

	
	
	{\it Data Structure for $D_t$}: Much more care is needed with how $D_t$ is stored and accessed.  Initially $D_0 = \mathcal{C}$, which contains an exponential number of polymers. The operations we need to perform on $D_t$ include deleting all polymers containing a single vertex $v$ (Step~\ref{item:delete}), checking whether a polymer is compatible with $D_t$ (Step~\ref{item:addtoD}), adding a polymer to $D_t$ (Step~\ref{item:addtoD}), and checking whether a polymer $\gamma$ is compatible with $D_t\setminus \{\gamma\}$ and if so deleting $\gamma$ from $D_t$ (Step~\ref{item:addtoD}). 
	
	The key observation is that once a particular vertex $v$ of $G$ is selected in Step\ref{item:pickv} and the coin flip is such that Step~\ref{item:delete} is performed, no polymers containing $v$ remain in $\mathcal{C}$ and the size of $D_t$ has been reduced dramatically. Because of this, it makes sense to keep track of which vertices in $D_t$ have been the subject of a deletion in Step~\ref{item:delete} at least once.  We let $D^*_t$ be (the state at time $t$) of an array of length $n$ with an entry corresponding to each vertex $v$ of $G$, where $D^*_t(v) = 0$ if Step~\ref{item:delete}, deleting $\mathcal{C}_v$ from $D_t$, has been performed for $v$ at least once, and $D^*_t(v) = 1$ if it has not.  We also let $N^*_t$ be the number of $1$'s in $D^*_t$. Initially, $D^*_0(v) = 1$ for all $v$ and $N^*_0 = n$.  At any time step, we know any polymer $\gamma$  where $D^*_t(v) = 1$ for all $v \in \gamma$ is in $D_t$. Because of this, initially $D^*_0$ completely describes~$D_0$. 
	
	Amidst these deletions, polymers are also added to $D_t$, and a polymer may contain both vertices where $D^*_t(v) = 0$ and vertices where $D^*_t(v) = 1$. All polymers added to $D_t$ in Step~\ref{item:addtoD} will be stored separately from $D^*_t$. The data structure used here will have the same idea as that for $B_t$. There will be an array $\overline{D}_t$ with an entry for each vertex where $\overline{D}_t(v)$ gives the number of polymers containing~$v$ that are currently stored in $D_t$. We also keep track of $\overline{N}_t$, the total number of polymers that are currently in $D_t$.  Adding a polymer $\gamma$ to $D_t$ involves incrementing $N_t$, incrementing the corresponding entries in $\overline{D}_t$, and adding a doubly linked list $\mathcal{L}^\gamma$ connecting the vertices of the polymer together. However, because polymers in $D_t$ need not be compatible, there may be more than one polymer in $D_t$ containing a given vertex $v$. Because of this, instead of each vertex being in at most one doubly linked list $\mathcal{L}^\gamma$, it may be in many such doubly linked lists. To maintain all such pointers for vertex $v$ (two for each $\mathcal{L}^\gamma$ where $v \in \gamma$), these pointers are themselves stored in another doubly linked list $\mathcal{L}^v$.  
	
	The necessary operations that must be performed on $D_t$ can be implemented as follows. Note that initially $D^*_0 (v) = 1$ for all $v$, $N^*_0 = n$, $\overline{D}_0(v) = 0 $ for all $v$, $\overline{N}_0 = 0$, there are no doubly linked lists $\cL^\gamma$, and each doubly linked list $\cL^v$ is empty.  These steps are presented in a logical order for explaining the data structures for $D_t$, rather than in the order in which they occur in the description of the Polymer Dynamics Bounding Chain. 
	
	\begin{itemize} 
		\item {\it Adding polymer $\gamma$ to $D_t$ (Step~\ref{item:addtoD})}: For each $v$ in $\gamma$, increment $\overline{D}_t(v)$. Increment $\overline{N}_t$.  Create a doubly linked list $\cL^\gamma$, which has a node for each $v \in \gamma$ that notes the color $v$ is assigned in $\gamma$ and containing a forward pointer and a backward pointer. For each $v$, this node containing the pointers of $\mathcal{L}^\gamma$ is inserted at the front of the doubly linked list $\cL^v$. This takes time $O(|\gamma|)$. 
		\item {\it Deleting all polymers containing $v$ from $D_t$ (Step~\ref{item:delete})}: If $D^*_t(v) = 1$, decrement $N^*_t$ and set $D^*_t (v) = 0$. If $\overline{D}_t(v) > 0$, set $\overline{N}_t \leftarrow \overline{N}_t - \overline{D}_t(v)$ and set $\overline{D}_t(v) = 0$. Additionally, for each node in $\cL^v$, we explore the corresponding linked list $\cL^\gamma$ using the pointers in the node.  We delete every other node of $\cL^\gamma$, possible in time $O(|\gamma|)$ because for list $\cL^w$ also containing a node of $\cL^\gamma$, the nodes before and after this node in $\cL^w$ can easily be connected to each other because $\cL^w$ is doubly linked. For each node in $\cL^w$ that is deleted, we also decrement $\overline{D}_t(w)$. This takes time $O(\sum_{\gamma \ni v} |\gamma|)$. 
		
		\item {\it Checking if polymer $\gamma$ is compatible with $D_t$ (Step~\ref{item:addtoD})}: For each $w$ that is in or adjacent to $\gamma$, check whether $D^*_t(w) = 0$ and $\overline{D}_t(w) = 0$.  If both are 0 for all such $w$, then $\gamma$ is compatible with $D_t$.  If at least one is nonzero, then $\gamma$ is not compatible with $D_t$. This takes time $O(\Delta|\gamma|)$.
		
		\item {\it Checking whether a polymer $\gamma$ is compatible with $D_t\setminus \{\gamma\}$ and if so, deleting $\gamma$ from $D_t$ (Step~\ref{item:addtoB})}: We only need to check this case when $\gamma$ is compatible with $B_t$. We first check whether $\gamma$ is compatible with $D_t$, as described above.  If it is, we add $\gamma$ to $B_t$. If $\gamma$ is compatible with $D_t$ then it cannot be in $D_t$, so $D_t$ and $D_t \setminus \{\gamma\}$ are the same and no further steps are needed. 
		
		If $\gamma$ is incompatible with $D_t$, the next step is to check if it is also incompatible with $D_t \setminus \{\gamma\}$.  We check if $D^*(v) = 0$ for all $v$ in or adjacent to $\gamma$, $\overline{D}_t(v) = 0$ for all $v$ adjacent to $\gamma$, and $\overline{D}_t(v) = 1$ for all $v \in \gamma$; if at least one of these does not hold, $\gamma$ cannot be compatible with $D_t \setminus \{\gamma\}$ and so we are not in Step~\ref{item:addtoB} and we proceed to Step~\ref{item:addtoD}.  If all these hold, we check further to see whether all $\overline{D}_t(v) = 1$ for $v \in \gamma$ because of the presence of the single polymer $\gamma$ or due to the presence of other polymers, which we can detect by looking at the doubly linked list(s) connecting the vertices in $\gamma$, including the colors assigned to each vertex. If there are polymers other than $\gamma$ here, either multiple smaller polymers or a polymer with the same vertices as $\gamma$ but differently assigned colors, then $\gamma$ is incompatible with $D_t\setminus \{\gamma\}$ and we are not in Step~\ref{item:addtoB}. If instead we find exactly polymer $\gamma$, then $\gamma$ is compatible with $D_t \setminus \{\gamma \}$ and  we add $\gamma$ to $B_t$ and delete it from $D_t$ as previously described. 
		This takes time $O(\Delta|\gamma|)$. 
	\end{itemize}
	
	Because any polymer added to $D_t$ was drawn from the distribution $\nu_v$ for some $v$, as argued above its expected size is constant.  Therefore all of the implementations above, with the possible exception of deleting all polymers containing a single vertex, takes expected constant time. Setting $D^*_t(v) = 0$ takes constant time, but it may take longer to remove any polymers that were added as linked lists $\cL^\gamma$.  However, we can amortize the cost of deletion if we pay for the cost of deleting a polymer when we add it, as all doubly linked lists $\cL^\gamma$ must be added before they are deleted. Doing this amortization makes the cost of adding a polymer to $D_t$ be $O(|\gamma|) + O(|\gamma|) = O(|\gamma|)$, while the amortized cost of deleting all polymers containing a vertex $v$ is now $O(1)$. Thus all necessary operations for $D_t$ can be performed in amortized expected constant time. 
	
	Finally, we note that the termination condition $B_t = B_t \cup D_t$, equivalent to $D_t = \emptyset$, can be checked in constant time by verifying $N^*_t = 0$ and $\overline{N}_t = 0$. The first condition $N^*_t = 0$ verifies that every vertex $v$ has had all polymers containing $v$ deleted at least once, and the second condition $\overline{N}_t = 0$ verifies that there are currently no additional polymers in $D_t$. 
	
	The total space used for these data structures after $t$ steps of the algorithm is in expectation at most $O(n+t)$: $O(n)$ for $\overline{B}_t$, $D^*_t$ and $\overline{D}_t$, and $O(|\gamma|)$ for each polymer $\gamma$ added to $D_t$, which is constant in expectation. There are at most $t$ polymers in $D_t$, one added each step, so this is $O(t)$ space at most in expectation. 
\end{proof}

Together, Theorem~\ref{thmPolymerSample}, Lemma~\ref{lemCFTP}, and Lemma~\ref{lemDataStructure} imply Theorem~\ref{thmPolymerSample1}.

\section{Applications to low-temperature spin systems}
\label{secApplications}

We present in this section the details of the applications spin system sampling given in Section~\ref{sub:app:spin}.

\subsection{The hard-core model on unbalanced bipartite graphs}
\label{AppSpinApplications:hc}

We start by proving Corollary~\ref{corHCunbalanced} for which we use Theorem~\ref{thmPolymerSample1}.

\begin{proof}[Proof of Corollary~\ref{corHCunbalanced}]
	
	Let $G$ be an $n$-vertex bipartite graph with partite sets $L$ and $R$. 
	Suppose that the vertices in $L$ have maximum degree $\Delta_L$, and that the vertices in $R$ have maximum degree $\Delta_R$ and minimum degree $\delta_R$.
	
	We can sample from the hardcore model distribution $\mu_G^{hc}$ (see~\eqref{eq:hc:dist} for its the definition) by considering an auxiliary distribution $\overline{\mu}$, which is a distribution on subsets $S \subseteq R$.  For a subset $S \subseteq R$ with neighbors $N(S)$ in $L$, this distribution is given~by 
	\[
	\overline{\mu} (S) = \frac{(1 + \lambda)^{|L \setminus N(S)|} \cdot \lambda^{|S|}} {Z},
	\]
	where $Z = \sum_{S \subseteq R} (1 + \lambda)^{|L \setminus N(S)|} \cdot \lambda^{|S|}$. 
	One can sample from $\mu_G^{hc}$ by first sampling $S\subseteq R$ according to $\overline{\mu}$, adding $S$ to the independent set, and then for each vertex $v \in L \setminus N(S)$ include $v$ in the independent set with probability $\lambda/(1+\lambda)$ independently. This results in exactly the desired distribution $\mu_G^{hc}$.
	
	To sample from $\overline{\mu}$ we use a polymer model, which we define next.
	Let $R^2$ be the graph whose vertices are the vertices of $R$, where two vertices are adjacent if they are at distance two in $G$. Note the maximum degree in $R^2$ is $\Delta_R(\Delta_L-1)$. We define a polymer $\gamma$ to be a connected subset of $R^2$, its neighborhood $N(\gamma)$ to be all vertices in $L$ adjacent to a vertex in $\gamma$, and its weight as
	\[
	w_\gamma = \frac{\lambda^{|\gamma|}}{(1+\lambda)^{|N(\gamma)|}} \leq \left(\frac{\lambda}{(1+\lambda)^{\delta_R/\Delta_L}}\right)^{|\gamma|}.
	\]
	Polymers will not be labelled; that is, $q = 1$. 
	Compatible polymer configurations in $R^2$ correspond exactly to subsets $S \subseteq R$, and the weights $w_\gamma$ mean a compatible polymer configuration corresponding to subset $S$ has probability exactly $\overline{\mu}(S)$. 
	
	To sample from this subset polymer model in $R^2$, we show that when condition~\eqref{eqn:unbalancedcond} holds,
	the conditions~\eqref{eqPolyLamCondition} and~\eqref{eqPerfectSamplingCondition} of Theorem~\ref{thmPolymerSample1} hold, implying the existence of a perfect sampling algorithm with expected running time $O(n \log n)$. 
	
	First, when \eqref{eqn:unbalancedcond} holds, it follows that $e\lambda(\Delta_L-1) \Delta_R \leq (1+\lambda)^{\delta_R/\Delta_L}$. Then 
	\[
	w_\gamma 
	\leq \left(\frac{\lambda}{(1+\lambda)^{\delta_R/\Delta_L}}\right)^{|\gamma|} \leq \left(\frac{1}{e (\Delta_L-1) \Delta_R}  \right)^{|\gamma|},
	\]
	and since $\Delta_R (\Delta_L-1)$ is the maximum degree in the host graph $R^2$, condition~\eqref{eqPolyLamCondition} holds. 
	
	From (\ref{eqn:unbalancedcond}) it also follows that 
	\begin{equation*}
		\frac{e \lambda \Delta_R (\Delta_L -1)}{(1+\lambda)^{\delta_R/\Delta_L}} < \frac{e \Delta_R (\Delta_L -1)}{1 + (1+e) \Delta_R (\Delta_L - 1)}.
	\end{equation*}
	Because all quantities in this equation are constants that do not depend on the size of the graph, we conclude there exists $\theta \in (0,1)$ such that 
	\begin{equation}\label{eq:hc_theta}
		\frac{e \lambda \Delta_R (\Delta_L -1)}{(1+\lambda)^{\delta_R/\Delta_L}} \leq \theta \frac{e \Delta_R (\Delta_L -1)}{1 + (1+e) \Delta_R (\Delta_L - 1)}.
	\end{equation}
	
	\noindent Using Lemma~\ref{lemma:Tvk} (which precisely counts the number of rooted subtrees of size $k$ of the $\Delta$-regular tree), the number of polymers of size $k$ that are incompatible with $v$ is at most $(\Delta_R (\Delta_L -1) + 1) (e \Delta_R (\Delta_L-1))^{k-1}/k$; recall that $\Delta_R (\Delta_L - 1)$ is the maximum degree of $R^2$. It follows from this and \eqref{eq:hc_theta} that 
	\begin{align*}
		\sum_{\gamma \not\sim {v}} |\gamma| w_\gamma &\leq \sum_{k = 1}^\infty k \cdot \frac{(\Delta_R (\Delta_L -1) + 1) (e \Delta_R (\Delta_L-1))^{k-1}}{k}  \cdot \left(\frac{\lambda}{(1+\lambda)^{\delta_R/\Delta_L}}\right)^{k} \\
		&= \frac{\Delta_R (\Delta_L -1 ) + 1}{e \Delta_R (\Delta_L -1)} \sum_{k = 1}^\infty \left(\frac{e \lambda \Delta_R (\Delta_L-1)}{(1+\lambda)^{\delta_R/\Delta_L}} \right)^k \\
		&\leq\frac{\Delta_R (\Delta_L -1 ) + 1}{e \Delta_R (\Delta_L -1)} \sum_{k = 1}^\infty \theta \left( \frac{e \Delta_R (\Delta_L -1)}{1 + (1+e) \Delta_R (\Delta_L - 1)} \right)^k 
		\\&= \theta \cdot \frac{\Delta_R (\Delta_L -1 ) + 1}{e \Delta_R (\Delta_L -1)} \frac{\frac{e \Delta_R (\Delta_L -1)}{1 + (1+e) \Delta_R (\Delta_L - 1)}}{1- \frac{e \Delta_R (\Delta_L -1)}{1 + (1+e) \Delta_R (\Delta_L - 1)}}
		\\&= \theta.
	\end{align*}
	Thus, condition~\eqref{eqPerfectSamplingCondition} also holds, and Theorem~\ref{thmPolymerSample1} supplies the perfect sampling algorithm with the desired running time.
\end{proof}

\subsection{Potts Model on Expander Graphs}
\label{AppSpinApplications:potts}

The $Q$-color Potts model on $G$ at inverse temperature $\beta$ is a distribution over all (not necessarily proper) $Q$-colorings of $V$. Let $\Omega_{G,Q}$ be all colorings $\omega: V \rightarrow [Q]$.  For a coloring $\omega \in \Omega_{G, Q}$ with $m(G, \omega)$ monochromatic edges, this distribution has 
\[
\mu_G ^{\mathrm{Potts}}
(\omega) = \frac{e^{\beta m(G, \omega)}}{Z},
\]
where $Z = \sum_{\Omega_{G,q}} e^{\beta m(G, \omega)}$.  

For $j \in [Q]$, let $\Omega_{G, Q, j}$ be all colorings in $\Omega_{G, Q}$ such that strictly more than $n/2$ vertices are assigned color $j$.
Let $\overline{\Omega}_{G,Q} := \bigcup_{j \in [Q]}  \Omega_{G,Q,j}$ and consider the distribution
$\overline{\mu} ^{\mathrm{Potts}}$ given by:
\[
\overline{\mu} ^{\mathrm{Potts}}(\omega) = \frac{e^{\beta m(G, \omega)}}{\hat Z} \mathbbm{1}(\omega \in \overline{\Omega}_{G,Q}), \text{ where } \hat Z = \sum_{\omega \in \overline{\Omega}_{G,Q}} e^{\beta m(G, \omega)}.
\]

It follows from \cite{jenssen2020algorithms} that, under suitable conditions, (as we detail next) a perfect sample from $\overline{\mu} ^{\mathrm{Potts}}$ is an $e^{-n}$-approximate sample from $\mu ^{\mathrm{Potts}}$.  

Let $\overline{\mu} ^{\mathrm{Potts}}_{j}$ be the distribution $\overline{\mu} ^{\mathrm{Potts}}$ conditioned on being close to the ground state that is entirely color $j$, that is, conditioned on being in $\Omega_{G,Q,j}$.  Because all $Q$ ground states are symmetric, one can sample from $\overline{\mu} ^{\mathrm{Potts}}$ by first picking a uniformly random $j \in [Q]$, and then sampling from  $\overline{\mu} ^{\mathrm{Potts}}_{j}$. Sampling from  $\overline{\mu} ^{\mathrm{Potts}}_{j}$ can be done using the polymer model we define next.

For $\omega \in \Omega_{G,Q,j}$, let $\Gamma(\omega) = \{v \in V: \omega(v) \neq j\}$. 
Consider a subset polymer model whose host graph is $G$ where a polymer is a graphlet of $G$ whose vertices have colors other than $j$ (there are thus $Q-1$ colors available to color the vertices of a polymer). 
For a polymer $\gamma$, we let $w_\gamma= \exp (- \beta B(\gamma))$  where $B(\gamma)$ is the number of bichromatic edges of the colored subgraph $\gamma$ plus the number of boundary edges of $\gamma$.

There is a bijection between Potts configurations in $\Omega_{G,Q,j}$ and compatible polymer configurations consisting of at most $n/2$ total vertices, where vertices assigned a color other than $j$ are identified. A sampling algorithm for this polymer model can give a sampling algorithm for $\overline{\mu} ^{\mathrm{Potts}}_j$, where if the polymer configuration produced has more than $n/2$ total vertices, that configuration is rejected and resampling occurs. As we will see, the probability of needing to resample is small. 

Using this polymer model representation,~\cite{jenssen2020algorithms} gives an efficient $\varepsilon$-approximate sampling algorithm with for $\mu_G ^{\mathrm{Potts}}$ , derived from an approximate sampling algorithm for $\overline{\mu} ^{\mathrm{Potts}}_j$. This algorithm applies to all $\alpha$-expanding graphs $G$ with maximum degree $\Delta$ whenever $\alpha > 0$, $\Delta \geq 3$, $Q \geq 2$, and $\beta > 4 \log((Q-1)\Delta)/\alpha$. However, it involves a polymer enumeration step, and so its runtime is (omitting the dependence on other parameters) of the form $n^{O(\log \Delta)}$.

In~\cite{chen2021fast}, the authors take a Markov chain approach and give an $\varepsilon$-approximate sampling algorithm for $\mu ^{\mathrm{Potts}}$, again via approximate sampling from $\overline{\mu} ^{\mathrm{Potts}}_j$, on $\alpha$-expanding graphs $G$ with maximum degree $\Delta$  whenever $\beta \geq \frac{5+3\log((Q-1)\Delta)}{\alpha}$. This has running time $O(n \log(n/\varepsilon) \log(1/\varepsilon))$. In both of these prior works, it must hold that $\varepsilon \geq Qe^{-n}$.

Our algorithm gives an even larger range of $\beta$ in which $O(n \log n)$ sampling is possible, and removes the dependence on $\varepsilon$.  It produces an exact sample for $\overline{\mu} ^{\mathrm{Potts}}$ rather than an approximate sample, although this sample is still only an $e^{-n}$-approximate sample for $\mu_G ^{\mathrm{Potts}}$.

\begin{proof}[Proof of Corollary~\ref{corPottsEx}]
	From the discussion above, it suffices to generate a perfect sample from $\overline{\mu} ^{\mathrm{Potts}}_{j}$ for any $j \in [Q]$ using the subset polymer model described above. We will show that conditions~\eqref{eqPolyLamCondition} and~\eqref{eqPerfectSamplingCondition} of Theorem~\ref{thmPolymerSample1} hold, implying the existence of a perfect sampling algorithm for $\overline{\mu} ^{\mathrm{Potts}}_{j}$ with expected running time $O(n \log n)$.
	
	When $\beta \geq \frac{1 + \log( (Q-1)\Delta)}{\alpha}$, since $G$ is an $\alpha$ expander, $w_\gamma = \exp (- \beta B(\gamma)) \leq\left( \frac{1}{e\Delta(Q-1)}\right)^{|\gamma|} \leq \lambda^{|\gamma|}$ for a suitable $\lambda < \lambda_*(\Delta,Q-1)$, and thus condition~\eqref{eqPolyLamCondition} is met. 
	Moreover,
	when (\ref{eqn:pottsbeta}) holds, then rearranging terms shows
	\begin{align*}
		e(Q-1)\Delta e^{-\alpha \beta} < \frac{e\Delta}{e\Delta + \Delta + 1}.
	\end{align*}
	Because all quantities in this equation are constants that do not depend on the size of the graph, we conclude there exists $\theta \in (0,1)$ such that 
	\begin{equation*}
		e(Q-1)\Delta e^{-\alpha \beta} \leq \theta \cdot  \frac{e\Delta}{e\Delta + \Delta + 1}.
	\end{equation*}
	
	Using the bound $(e\Delta)^{k-1}/k$ for the number of graphlets of size $k$ containing a given vertex, we deduce
	that the number of polymers incompatible with $v$ of size $k$ is at most $(\Delta+1) (e\Delta)^{k-1} (Q-1)^k / k$.
	(This bound for the number of graphlets appears  
	in~\cite{borgs1989unified} but can also be deduced by direct computation from Lemma~\ref{lemma:Tvk} below which provides a tighter bound.)
	Therefore, we get
	\begin{align*}
		\sum_{\gamma \not\sim {v}} |\gamma| w_\gamma &\leq \sum_{k = 1}^\infty k \cdot \frac{(\Delta+1) (e\Delta)^{k-1} (Q-1)^k}{k} e^{-\beta \alpha k} = \frac{\Delta + 1}{e\Delta} \sum_{k = 1}^\infty  \left( e \Delta (Q-1)e^{-\beta \alpha} \right)^k \\
		& \leq \frac{\Delta + 1}{e\Delta} \sum_{k = 1}^\infty \theta \cdot \left( \frac{e\Delta}{e\Delta + \Delta + 1}\right)^k = \theta \cdot \frac{\Delta + 1}{e\Delta} \left( \frac{\frac{e\Delta}{e\Delta + \Delta + 1}}{1-\frac{e\Delta}{e\Delta + \Delta + 1}}
		\right) = \theta.
	\end{align*}
	Thus, condition~\eqref{eqPerfectSamplingCondition} holds, and Theorem~\ref{thmPolymerSample1} provides a perfect sampling algorithm with the desired running time.
\end{proof}

\section{Graphlet sampling: hardness}
\label{sec:hardness}

In this section we establish the sharpness of the threshold $\lambda_*(\Delta,q)$. In particular, we prove
Lemmas~\ref{lemma:comp:hard} and~\ref{lemma:part:finite} from the introduction.
We first prove the following more general variant of Lemma~\ref{lemma:part:finite}, which corresponds to the $q=1$ case.

\begin{lemma}
	\label{lemma:part:finite-general}
	The partition function $Z_{G,r,\lambda}$ is finite for every (possibly infinite) graph $G=(V,E)$
	of maximum degree $\Delta$
	and every $r \in V$ if and only if $\lambda \le \lambda_*(\Delta,q)$.
	Moreover, $Z_{G,r,\lambda} \le 39$ when $\lambda \le \lambda_*(\Delta,q)$.
\end{lemma}

The following combinatorial fact will be used in the proof of Lemma~\ref{lemma:part:finite-general}.

\begin{lemma} 
	\label{lemma:Tvk}
	Consider the infinite $\Delta$-regular tree rooted at $\rho$, and
	let $T_{k}$ denote the number of subtrees of size $k$ rooted at $\rho$. Then for $k \ge 1$:
	$$
	T_{k} = \frac{\Delta}{(\Delta-1)k+1} \binom{(\Delta-1)k+1}{k-1}.
	$$ 
\end{lemma}
\begin{proof}
	For positive integers $a$ and $b$, let
	$$
	A_k(a,b) = \frac{a}{bk+a}\binom{bk+a}{k}.
	$$
	The numbers $A_k(a,b)$ are a generalization of the Catalan numbers and satisfy the recurrence
	\begin{equation}
		\label{eq:recurrence}
		A_k(a+c,b) = \sum_{i=0}^{k} A_i(a,b)A_{k-i}(c,b);
	\end{equation}
	see~\cite{Kahkeshani}. It is known (see, e.g., Lemma 1 in~\cite{read2021hardness}) that the number of subtrees of size $k$ containing the root 
	of the infinite $(\Delta-1)$-ary tree is 
	$$
	\frac{1}{(\Delta-2)k+1} \binom{(\Delta-1)k}{k} = \frac{1}{(\Delta-1)k+1} \binom{(\Delta-1)k+1}{k},
	$$
	which is also equal to $A_k(1,\Delta-1)$.
	Hence, for $k \ge 2$, we have
	\begin{align*}
		T_k &= \sum_{k_1=0}^{k-1} \sum_{k_2=0}^{k-1-k_1} \cdots \sum_{k_{\Delta-1}=0}^{k-1-(k_1+\dots+k_{\Delta-2})} A_{k-1-(k_1+\dots+k_{\Delta-1})}(1,\Delta-1) \cdot \prod_{i=1}^{\Delta-1} A_{k_i}(1,\Delta-1).   
	\end{align*}
	Using~\eqref{eq:recurrence}, we see that
	$$
	\sum_{k_{\Delta-1}=0}^{k-1-(k_1+\dots+k_{\Delta-2})} A_{k-1-(k_1+\dots+k_{\Delta-1})}(1,\Delta-1) A_{k_{\Delta-1}}(1,\Delta-1) = A_{k-1-(k_1+\dots+k_{\Delta-2})}(2,\Delta-1),
	$$
	and using~\eqref{eq:recurrence} repeatedly we get for $k \ge 2$
	\begin{align*}
		T_k = A_{k-1}(\Delta,\Delta-1) &= \frac{\Delta}{(\Delta-1)(k-1)+\Delta}\binom{(\Delta-1)(k-1)+\Delta}{k-1} \\
		&= \frac{\Delta}{(\Delta-1)k+1}\binom{(\Delta-1)k+1}{k-1}.
	\end{align*}
	The claim also holds trivially for $k=1$ and the result follows.
\end{proof}

We are now ready to provide the proof of Lemma~\ref{lemma:part:finite-general}.

\begin{proof}[Proof of Lemma~\ref{lemma:part:finite-general}] 
	Let $C_{v,k}$ be the number of graphlets of $G$ with $k$ vertices that contain $v$. Then, since $f \le 1$ by assumption, we have
	$$
	Z_{G,r,\lambda} = \sum_{\gamma \in \mathcal S (G,r,q)\cup\{\emptyset\}} \lambda^{|\gamma|} f(\gamma) \le \sum_{k \ge 0} \, \sum_{{\gamma \in \mathcal S (G,r,q)\cup\{\emptyset\}}: |\gamma| = k} \lambda^k  = \sum_{k \ge 0} C_{v,k} \cdot \lambda^k q^k.
	$$	
	
	Consider the infinite $\Delta$-regular tree rooted at $\rho$, and
	let $T_{k}$ be the number of subtrees of size $k$ rooted at $\rho$ so that $C_{v,k} \le T_{k}$.
	(The latter bound captures the key idea in the proof: the infinite $\Delta$-regular tree is the worst case among graphs of maximum degree $\Delta$
	and the the threshold $\lambda_*(\Delta,q)$ arises as the threshold at which there is a change in the number of fixed points of the corresponding tree recurrence.)
	
	
	Then, it follows from Lemma~\ref{lemma:Tvk} that
	$$
	Z_{G,r,\lambda} \le 1 + \sum_{k \ge 1} \frac{\Delta}{(\Delta-1)k+1} \binom{(\Delta-1)k+1}{k-1} \lambda^kq^k.
	$$
	
	We apply the ratio test for the summation on the right hand side and consider the limit:	 
	\begin{align*}
		L 
		&:=\lim_{k \to \infty} \frac{\frac{\Delta}{(\Delta-1)(k+1)+1}\binom{(\Delta-1)(k+1)+1}{k} q^{k+1}  \lambda^{k+1}}{\frac{\Delta}{(\Delta-1)k+1}\binom{(\Delta-1)k+1}{k-1} q^k  \lambda^k} \\&= q\lambda 	\lim_{k \to \infty} \frac{(\Delta-1)k+1}{(\Delta-1)(k+1)+1} \frac{\binom{(\Delta-1)(k+1)+1}{k} }{\binom{(\Delta-1)k+1}{k-1}}. 	
	\end{align*}
	We have that 
	$$
	A_k := \frac{\binom{(\Delta-1)(k+1)+1}{k} }{\binom{(\Delta-1)k+1}{k-1}} = \frac{1}{k} \cdot 
	\frac{[(\Delta-1)(k+1)+1]!\, [(\Delta-2)k+2]!}
	{[(\Delta-2)k+\Delta]! \, [(\Delta-1)k+1]!}.
	$$
	Using the inequality
	$
	h(n) e^{1/(12n+1)} \le n! \le h(n) e^{1/12n}
	$
	where $h(n) = \sqrt{2\pi}n^{n+1/2}e^{-n}$, and setting
	$$
	B_k = \frac{h((\Delta-1)(k+1)+1) h((\Delta-2)k+2)}{h((\Delta-2)k+\Delta) h((\Delta-1)k+1)}
	$$
	we get
	\begin{align}
		\label{eq:asymp:AB}
		\frac{B_k}{k} \frac{e^{\frac{1}{12[(\Delta-1)(k+1)+1]+1}+\frac{1}{12[(\Delta-2)k+2]+1}}}{e^{\frac{1}{12[(\Delta-2)k+\Delta]}+\frac{1}{12[(\Delta-1)k+1]}}} \le A_k \le \frac{B_k}{k} \cdot \frac{e^{\frac{1}{12[(\Delta-1)(k+1)+1]}+\frac{1}{12[(\Delta-2)k+2]}}}{e^{\frac{1}{12[(\Delta-2)k+\Delta]+1}+\frac{1}{12[(\Delta-1)k+1]+1}}}.
	\end{align}
	From a direct calculation, it can be checked that
	\begin{align*}
		B_k =& \frac 1e \Big(1 + \frac{\Delta-1}{(\Delta-1)k+1}\Big)^{(\Delta-1)k+1} 
		\Big(1 - \frac{\Delta-2}{(\Delta-2)k+\Delta}\Big)^{(\Delta-2)k+\Delta}  \\
		&\times   \frac{[(\Delta-1)(k+1)+1]^{\Delta-1}}{[(\Delta-2)k+2]^{\Delta-2}} 
		\frac{[(\Delta-2)k+2]^{1/2}}{[(\Delta-1)k+1]^{1/2}}
		\frac {[(\Delta-1)(k+1)+1]^{1/2}}{[(\Delta-2)k+\Delta]^{1/2}}.
	\end{align*}
	Hence, 
	$$
	\lim_{k \to \infty} \frac{B_k}{k} = \frac 1e e^{\Delta-1} \frac{1}{e^{\Delta-2}} \frac{(\Delta-1)^{\Delta-1}}{(\Delta-2)^{\Delta-2}} = \frac{(\Delta-1)^{\Delta-1}}{(\Delta-2)^{\Delta-2}} .
	$$
	From this and~\eqref{eq:asymp:AB} we deduce that $\lim_{k \to \infty} A_k = \frac{(\Delta-1)^{\Delta-1}}{(\Delta-2)^{\Delta-2}}$, which implies that 
	$$
	L = q\lambda \frac{(\Delta-1)^{\Delta-1}}{(\Delta-2)^{\Delta-2}}.
	$$
	From the ratio test we can then conclude that the series converges when $\lambda <  \lambda_*(\Delta,q)$
	and diverges when $\lambda > \lambda_*(\Delta,q)$.	
	
	It remains for us to consider the $\lambda = \lambda_*(\Delta,q)$	case, where
	$$
	Z_{G,r,\lambda} \le A + \sum_{k \ge 2} \frac{\Delta}{(\Delta-1)k+1} \binom{(\Delta-1)k+1}{k-1} \frac{(\Delta-2)^{(\Delta-2)k}}{(\Delta-1)^{(\Delta-1)k}}
	$$
	with $A =1 + \frac{(\Delta-2)^{(\Delta-2)}}{(\Delta-1)^{(\Delta-1)}} \le 5/4$.
	We also have for $k \ge 2$ that
	\begin{align*}	
		\binom{(\Delta-1)k+1}{k-1} 
		&\le \frac{h((\Delta-1)k+1)}{h(k-1)h((\Delta-2)k+2)} \times \frac{e^{\frac{1}{12((\Delta-1)k+1)}}}{e^{\frac{1}{12(k-1)+1}}e^{\frac{1}{12((\Delta-2)k+2)+1}}} \\
		&\le \frac{h((\Delta-1)k+1)}{h(k-1)h((\Delta-2)k+2)}.
	\end{align*}
	Therefore,
	\begin{align*}
		&Z_{G,r,\lambda} 
		\le \frac 54 + \sum_{k \ge 2} \frac{\Delta}{(\Delta-1)k+1}  \frac{h((\Delta-1)k+1)}{h(k-1)h((\Delta-2)k+2)} \frac{(\Delta-2)^{(\Delta-2)k}}{(\Delta-1)^{(\Delta-1)k}}\\
		&= \frac 54+\frac{1}{\sqrt{2\pi}}  \sum_{k \ge 2} \frac{\Delta}{(\Delta-1)k+1}  \frac{[(\Delta-1)k+1]^{(\Delta-1)k+3/2}}{(k-1)^{k-1/2} [(\Delta-2)k+2]^{(\Delta-2)k+2+1/2}} \frac{(\Delta-2)^{(\Delta-2)k}}{(\Delta-1)^{(\Delta-1)k}}\\
		&= \frac 54 + \frac{1}{\sqrt{2\pi}} \sum_{k \ge 2}  \frac{\Delta[(\Delta-1)k+1]^{1/2}(k-1)^{1/2}}{[(\Delta-2)k+2]^{5/2}} \left[\frac{((\Delta-1)k+1)^{\Delta-1}(\Delta-2)^{\Delta-2}}{(k-1)((\Delta-2)k+2)^{\Delta-2}(\Delta-1)^{\Delta-1}}\right]^k.
	\end{align*}
	
	Now, letting $L_1(\Delta,k) =  \frac{\Delta[(\Delta-1)k+1]^{1/2}(k-1)^{1/2}}{[(\Delta-2)k+2]^{5/2}}$,	
	$L_2(\Delta,k) = \frac{((\Delta-1)k+1)^{\Delta-1}(\Delta-2)^{\Delta-2}}{(k-1)((\Delta-2)k+2)^{\Delta-2}(\Delta-1)^{\Delta-1}}$, and noting that $L_1$ is a decreasing function of $\Delta$, we have
	$$
	L_1(\Delta,k) \le \frac{3(2k+1)^{1/2}(k-1)^{1/2}}{(k+2)^{5/2}} \le \frac{3\sqrt{2}(k+2)^{1/2}(k+2)^{1/2}}{(k+2)^{5/2}} \le \frac{3\sqrt{2}}{(k+2)^{3/2}},
	$$
	and, 	
	\begin{align*}
		L_2(\Delta,k) &= \left(1 -\frac{\Delta}{(\Delta-1)(\Delta-2)k+2(\Delta-1)}\right)^{\Delta-2} \left(1+\frac{\Delta}{(\Delta-1)(k-1)}\right)\\
		&\le \exp\left[ \frac{\Delta}{\Delta-1}\left(-\frac{\Delta-2}{(\Delta-2)k+2} + \frac{1}{k-1}\right) \right] \\&\le \exp\left[ \frac{\Delta}{(\Delta-1)(k-1)}\right] \le \exp\left[ \frac{3}{2(k-1)}\right].
	\end{align*}
	Hence,
	\begin{align*}
		Z_{G,r,\lambda}&\le \frac 54 +  \frac{3}{\sqrt{\pi}} \sum_{k \ge 2} \frac{1}{(k+2)^{3/2}} \exp\left[ \frac{3 k}{2(k-1)}\right]
		\le \frac 54 + \frac{3e^{3}}{\sqrt{\pi}} \sum_{k \ge 2} \frac{1}{(k+2)^{3/2}} \le 39
	\end{align*}
	when $\lambda = \lambda_*(\Delta,q)$.
\end{proof}

We proceed next with the proof of Lemma~\ref{lemma:comp:hard}.
Let us first formally define the notion of a polynomial-time approximate sampler.

\begin{definition}
	\label{def:approx:sampler}
	An algorithm $\mathcal A$ that takes as input a graph $G=(V,E)$, $\lambda > 0$, a vertex $r \in V$, and an accuracy parameter $\varepsilon \in (0, 1]$
	is a polynomial-time approximate sampler for $\nu_{G,r,\lambda}$
	if it returns a sample from a distribution $\mu_{\mathcal A}$ such that $\|\mu_{\mathcal A} - \nu_{G,r,\lambda}\|_{\textsc{tv}} \le \varepsilon$
	and has a running time that is polynomial in $|V|$ and $1/\varepsilon$.
	A polynomial-time approximate sampler for $\nu_{G,\lambda}$ is defined analogously.
\end{definition}

We will require the following result.

\begin{lemma} 
	\label{lemma:sampler}
	Fix $\lambda < 1$, and let $q=1$ and $f=1$.  There is an algorithm with  the following guarantees.
	The algorithm takes as input a finite graph $G = (V,E)$, a vertex $v \in V$, and parameters $\eps, \delta \in (0,1)$.  
	The algorithm has access to an exact sampling algorithm for $\nu_{G,v,\lambda}$.  With probability at least $1-\delta$ the algorithm outputs a value 
	$\tilde{Z}_{G,v,\lambda}$ 
	such that
	$
	(1-\varepsilon)\tilde{Z}_{G,v,\lambda} \le	Z_{G,v,\lambda} \le (1+\varepsilon)\tilde{Z}_{G,v,\lambda}
	$.
	The algorithm's sample complexity (queries to  $\nu_{G,v,\lambda}$) and running time are $\poly(|V|,\varepsilon,\log(1/\delta))$.
\end{lemma}
\begin{proof}
	Let $u_1$, $u_2$, $\dots,u_{n-1}$ be the vertices in $V \setminus \{v\}$ in a fixed order we define later. Let $G_i$ be the graph that results from removing vertices $u_1,\dots,u_i$ from $G$ and let $Z_i = Z_{G_i,v,\lambda}$. 
	We let $Z_0 = Z_{G,v,\lambda}$ and note that $Z_{n-1} = 1+\lambda$.
	Then:
	$$
	\frac{1+\lambda} {Z_0}= \prod_{i=0}^{n-2} \frac{Z_{i+1}}{Z_{i}}.
	$$
	Let	$p_i = \frac{Z_{i+1}}{Z_i}$. For each $p_i$, we will find $\tilde p_i$ such that with probability at least $1 - \delta/n$:
	\begin{equation}
		\label{eq:approx:pi:mult}
		(1-\frac{\varepsilon}{4n}) \tilde p_i \le p_i \le (1+\frac{\varepsilon}{4n}) \tilde p_i.
	\end{equation}
	This implies that 
	$$(1-\frac{\varepsilon}{2n}) \frac{1}{\tilde p_i} \le \frac{1}{p_i} \le (1+\frac{\varepsilon}{2n}) \frac{1}{\tilde p_i},$$
	and so letting $\tilde Z_{G,v,\lambda} = (1+\lambda) \prod_{i=0}^{n-2} \frac{1}{\tilde p_i}$, we get that 
	$$
	(1-{\varepsilon}) \tilde Z_{G,v,\lambda} \le (1-\frac{\varepsilon}{2n})^n \tilde Z_{G,v,\lambda} \le {Z_0} \le (1+\frac{\varepsilon}{2n})^n \tilde Z_{G,v,\lambda} \le (1+\varepsilon) \tilde Z_{G,v,\lambda},
	$$
	with probability at least $1-\delta$ by a union bound.
	
	To obtain~\eqref{eq:approx:pi:mult}, we define the sequence of vertices $u_1,u_2,\dots,u_{n-1}$ such that $u_i$ is a leaf in the breadth-first search tree of $G_{i-1}$ rooted at $v$. This way, we can conveniently claim that $1 \ge p_i \ge \frac{1}{2}$. To see this, note that $S \in \mathcal S(G_{i+1})$ if and only if $S \in \mathcal S(G_{i}): u_{i+1} \not\in S$, so
	\begin{equation}
		\label{eq:zi}
		Z_{i+1} = \sum_{S \in \mathcal S(G_{i+1})} \lambda^{|S|} = \sum_{S \in \mathcal S(G_{i}): u_{i+1} \not\in S}  \lambda^{|S|}.
	\end{equation}
	Now, if $S \in \mathcal S(G_{i})$ and $S$ contains $u_{i+1}$,
	then, since $u_{i+1}$ is a leaf in the breadth-first search tree rooted at $v$, $S \setminus \{u_{i+1}\}$ is a connected subgraph that contains $v$. That is, every $S \in \mathcal S(G_{i})$ 
	that contains $u_{i+1}$ can be mapped to a unique 
	subgraph in $\mathcal S(G_{i})$ that does not contain $u_{i+1}$ of size $|S|-1$. Hence, since $\lambda < 1$,
	\begin{align*}
		\sum_{S \in \mathcal S(G_{i}): u_{i+1} \in S}  \lambda^{|S|} 
		\le \sum_{S \in \mathcal S(G_{i}): u_{i+1} \not\in S}  \lambda^{|S|},
	\end{align*}
	and so
	$$
	Z_i = \sum_{S \in \mathcal S(G_{i}): u_{i+1} \in S}  \lambda^{|S|} +\sum_{S \in \mathcal S(G_{i}): u_{i+1} \not\in S}  \lambda^{|S|} \le 2\sum_{S \in \mathcal S(G_{i}): u_{i+1} \not\in S}  \lambda^{|S|}.
	$$
	Combined with~\eqref{eq:zi}, this gives
	$
	p_i = \frac{Z_{i+1}}{Z_i} \ge  \frac{1}{2}.
	$
	Therefore, to deduce~\eqref{eq:approx:pi:mult}, it suffices to find $\tilde p_i$ such that
	$$
	\tilde p_i-\frac{\varepsilon}{16n} \le p_i \le \tilde p_i + \frac{\varepsilon}{16n}.
	$$
	For this, we draw $L$ samples from $\nu_{G_i,v,\lambda}$ and let $X_j$ be the indicator random variable for the event that the $j$-th sample contains $u_{{i+1}}$. Letting $\tilde p_i = \frac{1}{L} \sum_{i=1}^L X_j$, we get from a Chernoff bound that 
	$$
	\Pr[|\tilde p_i - p_i| \ge \rho p_i] \le 2 e^{-\frac{L \rho^ 2 p_i}{3}}
	$$
	and setting $\rho = \frac{\varepsilon}{16 n p_i}$ and $L = 384 \frac{n^2}{\varepsilon^2} \log (2n/\delta)$:
	$$
	\Pr\Big[|\tilde p_i - p_i| \ge \frac{\varepsilon}{16 n}\Big] \le 2 e^{-\frac{L \rho \varepsilon}{24n}} \le 2 e^{-\frac{L  \varepsilon^2}{384n^2}} = \frac{\delta}{n}.
	$$
	In summary, we have provided algorithm that computes $\tilde p_i$
	such that~\eqref{eq:approx:pi:mult} holds with probability at least $1 -\delta/n$.
	The algorithm has sample complexity and running time $\poly(n,1/\varepsilon,\log(1/\delta))$.
\end{proof}

We can now provide the proof of Lemma~\ref{lemma:comp:hard}.

\begin{proof}[Proof of Lemma~\ref{lemma:comp:hard}]
	It suffices to prove the result for $f=1$.
	Let $G = (V,E)$ be a finite graph of maximum degree $\Delta$, and 
	for ease of notation let $\lambda_* = \lambda_*(\Delta,1)$.
	We show that if there is a polynomial-time approximate sampler for $\nu_{G,v,\lambda}$ 
	for each $v \in V$ (see Definition~\ref{def:approx:sampler}),
	then there is also a polynomial-time approximate sampler for for $\nu_{G,\lambda}$.
	In~\cite{read2021hardness}, it was shown that there is no polynomial-time approximate sampler for
	for $\nu_{G,\lambda}$ when $\lambda \in (\lambda_*,1)$ unless NP=RP, and thus our result follows\footnote{We note that the hardness result in~\cite{read2021hardness} is stated for the harder problem of producing approximate samples with a running time depending polynomially on $\log(1/\varepsilon)$, but the proofs in~\cite{read2021hardness} extend to our setting without modification.}.
	
	First, 
	given exact sampling algorithms for $\nu_{G,v,\lambda}$
	for each $v \in V$ with $\poly(n)$ running times, and
	a desired accuracy parameter $\hat\varepsilon \in (0,1)$, 
	we provide a sampling algorithm whose output distribution is within $\hat\varepsilon$ total variation distance from $\nu_{G,\lambda}$. The running time of the algorithm will be 
	$\poly(n,1/\hat\varepsilon)$.
	We then modify the algorithm to use instead the polynomial-time approximate samplers for $\nu_{G,v,\lambda}$ (instead of the exact samplers) and show that this has a negligible effect on its running time and output distribution.
	
	The algorithm is the following:
	\begin{enumerate}
		\item Let $\varepsilon = \min\{\frac 18,\frac{\hat\varepsilon}{6} - \frac{1}{4n}\} $. For each $v \in V$, using the exact sampler for $\nu_{G,v,\lambda}$ and the algorithm from~Lemma~\ref{lemma:sampler}, obtain $\tilde{Z}_{G,v,\lambda}$ such that, with probability at least $1 - \frac{1}{n^2 2^n}$,
		\begin{equation}
			\label{eq:approx}
			(1-\varepsilon)\tilde{Z}_{G,v,\lambda} \le	Z_{G,v,\lambda} \le (1+\varepsilon)\tilde{Z}_{G,v,\lambda}
		\end{equation}
		
		\item Set $t=0$ and choose a vertex $v \in V$ with probability $\frac{\tilde{Z}_{G,v,\lambda}}{\sum_w \tilde{Z}_{G,w,\lambda}}$;
		\item Using the exact sampler for $\nu_{G,v,\lambda}$, draw a sample $S$ from $\nu_{G,v,\lambda}$;
		\item Accept and output $S$ with probability ${1}/{|S|}$;
		\item If $S$ is rejected, increase $t$, and if $t \le 2n^2$ go to Step 2 and repeat; otherwise output~$\emptyset$.
	\end{enumerate}
	
	First, we note that by Lemma~\ref{lemma:sampler}, Step 1 can be implemented in $\poly(n,1/\varepsilon)$.
	Steps 2 to 5 of the algorithm are repeated at most $O(n^2)$ times, and each of those step takes $\poly(n)$ time. Therefore, the overall running time of the algorithm is $\poly(n,1/\varepsilon) = \poly(n,1/\hat\varepsilon)$. Moreover,~\eqref{eq:approx} holds for every $v \in V$	with probability at least $1 - \frac{1}{n 2^n}$ by a union bound.
	
	Let $\dalgOverall$ be the output distribution of the algorithm,
	and let $\dalg$ be the output distribution when (i) Step 1 succeeds in finding the required approximations and (ii) a subgraph is accepted in Step 4 before $t > 2n^2$.
	We show next that $\|\dalg - \nu_{G,\lambda}\|_{\textsc{tv}} \le \hat \varepsilon$.
	Let $\tilde Z = \sum_{v \in V} \tilde{Z}_{G,v,\lambda}$
	and let $\mathcal S = \mathcal S(G) \cup \{\emptyset\}$. 
	The probability that the algorithm outputs $S \in \mathcal S$	in Step 4 in a give iteration~is:
	\begin{equation}
		\sum_{v \in S}\frac{\tilde{Z}_{G,v,\lambda}}{\tilde Z} \nu_{G,v,\lambda}(S) \frac{1}{|S|} 
		=: \frac{\phi(S)}{\tilde Z}. 
	\end{equation}	
	Then,
	$
	\dalg(S) = \frac{\phi(S)}{\sum_{\hat S \in \mathcal S} \phi(\hat S)}, 
	$
	and
	\begin{align}
		\|\dalg - \nu_{G,\lambda}\|_{\textsc{tv}} 
		&= \frac 12 \sum_{S \in \mathcal S} \Big|\frac{\phi(S)}{\sum_{W \in \mathcal S} \phi(W)} - \frac{\lambda^{|S|}}{ Z_{G,\lambda}}\Big|\notag =	\frac 12 \sum_{S \in \mathcal S} \Big|\frac{\sum_{v \in S} \nu_{G,v,\lambda}(S)\tilde Z_{G,v,\lambda}}{|S|\sum_{W \in \mathcal S} \phi(W)} - \frac{\lambda^{|S|}}{ Z_{G,\lambda}}\Big|\notag\\
		&\le \frac 12 \sum_{S \in \mathcal S}\sum_{v \in S} \Big|\frac{ \nu_{G,v,\lambda}(S)\tilde Z_{G,v,\lambda}}{|S|\sum_{W \in \mathcal S} \phi(W)} - \frac{\lambda^{|S|}}{|S| Z_{G,\lambda}}\Big|\notag\\
		&= \frac 12 \sum_{S \in \mathcal S}\sum_{v \in S} \frac{ Z_{G,v,\lambda}}{|S| Z_{G,\lambda}}\Big|\frac{ \nu_{G,v,\lambda}(S) \cdot \tilde Z_{G,v,\lambda}}{Z_{G,v,\lambda}} \cdot\frac{Z_{G,\lambda}}{\sum_{W \in \mathcal S} \phi(W)} - \frac{\lambda^{|S|}}{Z_{G,v,\lambda}}\Big|\notag\\
		&= \frac 12 \sum_{S \in \mathcal S}\sum_{v \in S} \frac{\nu_{G,v,\lambda}(S) \cdot Z_{G,v,\lambda}}{|S| Z_{G,\lambda}}\Big|\frac{ \tilde Z_{G,v,\lambda}}{Z_{G,v,\lambda}} \cdot\frac{Z_{G,\lambda}}{\sum_{W \in \mathcal S} \phi(W)} - 1\Big|.
		\label{eq:tv:bound}
	\end{align}
	From~\eqref{eq:approx} we get
	$$1-2\varepsilon \le \frac{1}{1+\varepsilon} \le \frac{\tilde Z_{G,v,\lambda}}{Z_{G,v,\lambda}} \le \frac{1}{1-\varepsilon} \le 1+2\varepsilon,$$
	since $\varepsilon \le 1/8$. 
	Moreover,
	\begin{align*}
		\sum_{W \in \mathcal S} \phi(W) &= 
		\sum_{W \in \mathcal S} \sum_{v \in W} \frac{\nu_{G,v,\lambda}(W)\tilde Z_{G,v,\lambda}}{|W|} \le (1+2\varepsilon)	\sum_{W \in \mathcal S} \sum_{v \in W} \frac{\lambda^{|W|}}{|W|} = (1+2\varepsilon)Z_{G,\lambda},
	\end{align*}
	and similarly  we can obtain that $\sum_{W \in \mathcal S} \phi(W) \ge (1-2\varepsilon)Z_{G,\lambda}$.
	Combining these bounds we get:
	$$
	\Big|\frac{ \tilde Z_{G,v,\lambda}}{Z_{G,v,\lambda}} \cdot\frac{Z_{G,\lambda}}{\sum_{W \in \mathcal S} \phi(W)} - 1\Big| \le \frac{4 \varepsilon}{1-2\varepsilon} \le 8 \varepsilon
	$$
	Plugging this into~\eqref{eq:tv:bound} we deduce that:
	\begin{align}
		\|\dalg - \nu_{G,\lambda}\|_{\textsc{tv}} \le 4 \varepsilon \sum_{S \in \mathcal S}\sum_{v \in S} \frac{\nu_{G,v,\lambda}(S) \cdot Z_{G,v,\lambda}}{|S|Z_{G,\lambda}}= \frac{4 \varepsilon}{Z_{G,\lambda}} \sum_{S \in \mathcal S}\sum_{v \in S} \frac{\lambda^{|S|}}{|S|} \le 4\varepsilon.\label{eq:eps:bound}
	\end{align}
	Let $\mathcal E_1$ be the event that Step 1 succeeds in finding the approximations, 
	let $\mathcal E_2$ the event that the algorithm accepts and outputs a graphlet in Step 4, and let $\mathcal E$ be the event that both $\mathcal E_1$ and $\mathcal E_2$ occur.
	Then, the triangle inequality and~\eqref{eq:eps:bound} imply:
	\begin{align}
		\|\dalgOverall - \nu_{G,\lambda}\|_{\textsc{tv}}  
		&\le
		\|\dalg - \nu_{G,\lambda}\|_{\textsc{tv}}+ \|\dalg - \dalgOverall\|_{\textsc{tv}} \le 4\varepsilon + \|\dalg - \dalgOverall\|_{\textsc{tv}}.\label{eq:tv:bound-1}
	\end{align}
	We proceed to bound $\|\dalg - \dalgOverall\|_{\textsc{tv}}$.
	\begin{align}
		\|\dalg - \dalgOverall\|_{\textsc{tv}} 
		&= 	\|\dalgOverall - \dalgOverall(\cdot \mid \mathcal E)\|_{\textsc{tv}} \le \dalgOverall(\neg\mathcal E) \notag\\
		&\le \dalgOverall(\neg\mathcal E_1) +\dalgOverall(\neg\mathcal E_2 \mid \mathcal E_1) \le \frac{1}{2n} + \dalgOverall(\neg\mathcal E_2\mid \mathcal E_1). \label{eq:alg:bound}
	\end{align}
	So it remains for us to bound $\dalgOverall(\neg\mathcal E_2\mid \mathcal E_1)$.	
	For this, note that the probability that a subgraph is accepted in Step 4 in an iteration is:
	\begin{align*}
		\sum_{S \in \mathcal S}\sum_{v \in S}\frac{\tilde{Z}_{G,v,\lambda}}{\tilde Z} \nu_{G,v,\lambda}(S) \frac{1}{|S|} 
		\ge \frac{1-2\varepsilon}{\tilde Z} 	\sum_{S \in \mathcal S}\sum_{v \in S} \frac{\lambda^{|S|}}{|S|} \ge \frac{(1-2\varepsilon)Z_{G,\lambda}}{\tilde Z}.
	\end{align*}
	Since
	$$
	\tilde Z = \sum_{v \in V} \tilde Z_{G,v,\lambda} \le (1+2\varepsilon)\sum_{v \in V} Z_{G,v,\lambda} \le (1+2\varepsilon)n Z_{G,\lambda}, 
	$$
	we get the acceptance probability in Step 4 in an iteration is at least 
	$$\frac{1-2\varepsilon}{(1+2\varepsilon)n} \ge \frac{1-4\varepsilon}{n} \ge \frac{1}{2n},$$
	provided $\mathcal E_1$ occurs.
	Therefore, if $X$ is geometric random variable with parameter $1/(2n)$, we have
	$$
	\dalgOverall(\neg\mathcal E_2\mid \mathcal E_1) \le \Pr[X \ge 2n^2] \le \Big(1 - \frac{1}{2n}\Big)^{2n^2} \le \frac{1}{e^n}, 	
	$$
	and plugging this bound into~\eqref{eq:alg:bound} and~\eqref{eq:tv:bound-1} we obtain:
	$$
	\|\dalgOverall - \nu_{G,\lambda}\|_{\textsc{tv}}  
	\le 4\varepsilon + \frac{3}{2n}.
	$$
	
	We have established so far the hardness of exactly sampling from $\nu_{G,v,\lambda}$ for $\lambda \in (\lambda_*,1)$. The same reduction (i.e., algorithm) with minor adjustments works for approximate sampling. Suppose that in Step 3 of the algorithm, we instead generate $T = 2n^2$ samples from a distribution $\mu_v$ such that $\|\mu_v - \nu_{G,v,\lambda}\|_{\textsc{tv}} \le \frac{\varepsilon}{n^2}$.	
	Let $\mu_{v}^{\otimes T}$ and $\nu_{G,v,\lambda}^{\otimes T}$ be the product distributions corresponding to $T$ independent samples from $\mu_{v}$ and $\nu_{G,v,\lambda}$ respectively.
	We have
	\begin{equation}
		\label{eq:tv:prod}
		\|\mu_{v}^{\otimes T} - \nu_{G,v,\lambda}^{\otimes T}\|_{\textsc{tv}} \leq T \|\mu_{v} - \nu_{G,v,\lambda}\|_{\textsc{tv}}  \leq  \frac{T \varepsilon}{n^2} \le 2\varepsilon.
	\end{equation}
	Let $\fdalg$ be the output distribution of the algorithm
	when using samples from $\mu_{v}$ in Step 3.
	Then:
	\begin{align*}
		\|\fdalg-\nu_{G,\lambda}\|_{\textsc{tv}} 
		&\le \|\fdalg-\dalgOverall\|_{\textsc{tv}}+\|\dalgOverall-\nu_{G,\lambda}\|_{\textsc{tv}} \le \|\fdalg-\dalgOverall\|_{\textsc{tv}}+4\varepsilon+\frac{3}{2n}.	
	\end{align*}
	Consider the following coupling between $\fdalg$ and $\dalgOverall$:
	use the same randomness for Steps 1, 2 and 5 and the optimal coupling for $\mu_{v}^{\otimes T}$ and $\nu_{G,v,\lambda}^{\otimes T}$ for Step 3. Then, from~\eqref{eq:tv:prod}, we get
	$$
	\|\fdalg-\nu_{G,\lambda}\|_{\textsc{tv}} 
	\le 6\varepsilon+\frac{3}{2n} \le \hat\varepsilon.$$
	That is, we obtain a polynomial-time approximate sampler for $\nu_{G,\lambda}$ which completes the reduction.
\end{proof}

\section*{Acknowledgements}

This work was carried out as part of the AIM SQuaRE workshop ``Connections between computational and physical phase transitions.''  We thank Tyler Helmuth, Alexandre Stauffer, and Izabella Stuhl for many helpful conversations.  WP is supported in part by NSF grant DMS-2309958 and CCF-2309708. SC is supported in part by NSF grants DMS-1803325 and CCF-2104795. AB is supported in part by NSF grant CCF-2143762. We also thank the anonymous referees for their many insightful corrections and suggestions.

\bibliographystyle{alpha}
\bibliography{subgraph}

\end{document}